\documentclass[11pt,letter]{article}
\usepackage{lineno}
\usepackage{amssymb}

\renewcommand{\tilde}{\widetilde}

\usepackage{fullpage}

\usepackage{microtype}
\usepackage{epsfig}
\usepackage{amssymb}
\usepackage{microtype}
\usepackage{wrapfig}
\usepackage{graphicx}
\usepackage{epsfig,amsmath,amsfonts,amssymb,amsthm,mathrsfs,ifpdf}
\usepackage{indentfirst,relsize}
\usepackage[numbers]{natbib}
\usepackage{setspace}
\usepackage{enumerate}
\usepackage{latexsym}
\usepackage{stackrel}
\usepackage[all]{xy}
\usepackage[usenames,dvipsnames]{pstricks}
\usepackage{pst-grad} 
\usepackage{pst-plot} 
\usepackage{xspace}
\usepackage{bbm}

\usepackage[super]{nth}
\usepackage{hyperref}

\usepackage[T1]{fontenc}
\usepackage{lineno}

\newcommand{\size}[1]{\left| #1 \right|}

\newcommand{\E}{\mathbb{E}}
\newcommand{\remove}[1]{}

\newcommand{\R}{\mathbb{R}}
\newcommand{\N}{\mathbb{N}}

\newcommand{\cS}{\mathcal{S}}

\newcommand{\cD}{\mathcal{D}}
\newcommand{\cE}{\mathcal{E}}

\newcommand{\cP}{\mathcal{P}}

\newcommand{\cR}{\mathcal{R}}

\newcommand{\Oh}{\mathcal{O}}
\newcommand{\cQ}{\mathcal{Q}}
\newcommand{\tOh}{\widetilde{{\mathcal O}}}

\newcommand{\supp}{\mbox{{\sc Supp}}}

\newcommand{\eps}{\varepsilon}
\newcommand{\pr}{\mathbb{P}}

\newcommand{\floor}[1]{{\lfloor{#1}\rfloor}}

\theoremstyle{plain}
\newtheorem{theo}{Theorem}[section]
\newtheorem{lem}[theo]{Lemma}

\newtheorem{cl}[theo]{Claim}
\theoremstyle{definition}
\newtheorem{defi}[theo]{Definition}
\newtheorem{rem}{Remark}
\newtheorem{obs}[theo]{Observation}

\newcommand{\high}{\mathsf{High}}

\title{Exploring the Gap between Tolerant and Non-tolerant\\
Distribution Testing\footnote{A preliminary version of this paper has been accepted in RANDOM 2022.}
}
\author{Sourav Chakraborty\footnote{Indian Statistical Institute, Kolkata, India. Email: chakraborty.sourav@gmail.com.}
\and
Eldar Fischer\footnote{Technion - Israel Institute of Technology, Israel. Email: eldar@cs.technion.ac.il.}
\and
Arijit Ghosh\footnote{Indian Statistical Institute, Kolkata, India. Email: arijitiitkgpster@gmail.com.}
\and
Gopinath Mishra\footnote{University of Warwick, UK. Email: gopianjan117@gmail.com. Research supported in part by the Centre for Discrete Mathematics and its Applications (DIMAP) and by EPSRC award EP/V01305X/1.}
\and
Sayantan Sen\footnote{Indian Statistical Institute, Kolkata, India. Email: sayantan789@gmail.com.}
}
\date{}

\begin{document}
\begin{titlepage}
\clearpage\thispagestyle{empty}
  \maketitle

\begin{abstract}
The framework of distribution testing is currently ubiquitous in the
field of property testing. In this model, the input is a probability
distribution accessible via independently drawn samples from an oracle.
The testing task is to distinguish a distribution that satisfies some
property from a distribution that is far in some distance measure from satisfying it. The task of tolerant testing imposes a further restriction, 
that distributions close to satisfying the property are also accepted.

This work focuses on the connection between the sample complexities of
non-tolerant testing of distributions and their tolerant testing counterparts. When limiting our scope to label-invariant (symmetric)
properties of distributions, we prove that the gap is at most quadratic, ignoring poly-logarithmic factors.
Conversely, the property of being the uniform distribution is indeed known
to have an almost-quadratic gap.

When moving to general, not necessarily label-invariant properties, the situation is more complicated, and we show some partial results. We
show that if a property requires the distributions to be
non-concentrated, that is, the probability mass of the distribution is sufficiently spread out, then it cannot be non-tolerantly tested with  $o(\sqrt{n})$ many samples, where $n$ denotes the universe size. Clearly, this
implies at most a quadratic gap, because a distribution can be learned
(and hence tolerantly tested against any property) using $\mathcal{O}(n)$ many samples. Being non-concentrated is a strong requirement on properties, as we also prove a close to linear lower bound against their tolerant tests.

Apart from the case where the distribution is non-concentrated, we also show if an input
distribution is very concentrated, in the sense that it is mostly
supported on a subset of size $s$ of the universe, then it can be learned
using only $\mathcal{O}(s)$ many samples. The learning procedure adapts to the input,
and works without knowing $s$ in advance.



\thispagestyle{empty}
\end{abstract}
\clearpage

\clearpage
\pagenumbering{arabic} 
\end{titlepage}

\section{Introduction}

Let $D$ be a distribution over a finite set $\Omega$, and $\mathcal{P}$ be a
property, that is, a set of distributions over $\Omega$. Given access to 
independent random samples from $\Omega$ according to the distribution $D$, 
we are interested in the problem of distinguishing whether the distribution $D$ is $\eta$-close to having the property $\mathcal{P}$, or is $\eps$-far from having 
the property $\mathcal{P}$, where $\eta$ and $\eps$ are two fixed proximity 
parameters such that $0\leq \eta<\eps \leq 2$. The distance of the 
distribution $D$ from the property $\mathcal{P}$ is defined as 
$\min\limits_{D' \in \mathcal{P}}\left| \left| D - D'\right|\right|_1$,
where $||D-D'||_1$ denotes the \emph{$\ell_1$-distance} between the distributions $D$ and $D'$~\footnote{Strictly speaking it is an infimum, but since all properties we consider are compact sets, it is equal to the minimum.}. A distribution $D$ is said to be $\eta$-close to $\cP$ if the distance of $D$ from $\cP$ is at most $\eta$. Similarly, $D$ is said to be $\eps$-far from $\cP$, if the distance of $D$ from $\cP$ is at least $\eps$. The goal is to design a tester that uses as few samples as possible. For $\eta>0$, the problem of distinguishing the two cases is referred to as the \emph{tolerant distribution testing} problem of $\mathcal{P}$, and the particular case where $\eta=0$ is referred to as the \emph{non-tolerant distribution testing} problem of $\mathcal{P}$. The sample complexity (tolerant and non-tolerant testing) is the number of samples required by the best algorithm that can distinguish with high probability (usually with probability at least $\frac{2}{3}$) whether the distribution $D$ is $\eta$-close to having the property $\mathcal{P}$, or is $\eps$-far from having the property $\mathcal{P}$.


While results and techniques from distribution testing are already interesting in their own right, they have also found numerous applications in central problems in Theoretical Computer Science, and in particular in property testing, e.g. graph isomorphism testing~\cite{fischer2008testing, goldreich2019testing} and function isomorphism testing~\cite{DBLP:journals/siamcomp/AlonBCGM13}, learning theory~\cite{ben2010theory, diakonikolas2017statistical, diakonikolas2016new}, and differential privacy~\cite{aliakbarpour2019private, gopi2020locally, zhang2021statistical,acharya2021inference}. Thus, understanding the tolerant and non-tolerant sample complexity of distribution testing is a central problem in theoretical computer science.

\color{black}

{There have been extensive studies of non-tolerant and tolerant testing of some specific distribution properties like {uniformity}, {identity with a fixed distribution}, {equality of two distributions} and {independence of a joint distribution}~\cite{Batu00, batu2001testing,paninski2008coincidence, valiant2011testing, ValiantV11, valiant2017automatic}. Various other specific distribution properties have also been studied~\cite{batu2017generalized,diakonikolas2017sharp}.} Then, some works investigated general tests for the large class of all shape-restricted properties of distributions, which contains properties like monotonicity, log-concavity, modality etc.~\cite{canonne2018testing, fischer2019improving}. This paper proves general results about the gap between tolerant and non-tolerant distribution testing that hold for large classes of properties.


\subsection{Our results }\label{sec:res}

We now informally present our results. The formal definitions are presented in Section~\ref{sec:pre}.
We assume that the distributions are supported over a set $\Omega=[n]=\{1,2,\ldots, n\}$. We first prove a result about label-invariant distribution properties (properties that are invariant under all permutations of $\Omega$). We show that, for any label-invariant distribution property, there is at most a quadratic blowup in its tolerant sample complexity as compared to its non-tolerant counterpart, ignoring poly-logarithmic factors.


\begin{theo}[Informal]\label{thm:symtolerant}
Any label-invariant distribution property that can be non-tolerantly tested using $\Lambda$ samples, can also be tolerantly tested using $\widetilde{\Oh}(\min\{\Lambda^2,n\})$ samples, where $n$ is the size of the support of the distribution~\footnote{$\tOh(\cdot)$ hides a poly-logarithmic factor.}.

\end{theo}


This result gives a unified way for obtaining tolerant testers from their non-tolerant counterparts. The above result will be stated and proved formally in Section~\ref{sec:symresult}. Moreover, in Section~\ref{sec:consttoltest}, we give a constructive variant of the tolerant tester of Theorem~\ref{thm:symtolerant}, when the property can be expressed as the feasible solution to a set of linear inequalities.

\begin{theo}[Informal]\label{thm:inf-conv}
Any label-invariant distribution property that can be non-tolerantly tested using $\Lambda$ samples and can be expressed as a feasible solution to $m$ linear inequalities, can also be tolerantly tested using $\widetilde{\Oh}(\min\{\Lambda^2,n\})$ samples and in time polynomial in $m$ and $n$, where $n$ is the size of the support of the distribution.
\end{theo}

\color{black}

Note that if $\Lambda= \Omega(\sqrt{n})$, Theorem~\ref{thm:symtolerant} is obvious.  It is only interesting if $\Lambda = o(\sqrt{n})$. Now we present a property for which this connection is useful. Consider a natural distribution property: given a distribution $D$ and a parameter $k$, we want to decide whether the support size of $D$ is at most $k$ or $\eps$-far from having support at most $k$. If $k = o(\sqrt{n})$, the query complexity for testing this problem is $\Oh(\frac{k}{\log k})$~\cite{valiant2017estimating}.

It is a natural question to investigate the extent to which the above theorem can be generalized. Though we are not resolving this question completely, as a first step in the direction of extending the above theorem for properties that are not necessarily label-invariant, we consider the notion of  \emph{non-concentrated} properties. By the notion of a non-concentrated distribution, intuitively, we mean that there is no significant portion of the base set of the distribution that carries only a negligible weight, making the probability mass of the distribution well distributed among its indices. Specifically, any subset $X \subseteq [n]$, for which $\size{X}$ is above some threshold (say $\beta n$ with $\beta \in (0, \frac{1}{2})$), has probability mass of at least another threshold (say $\alpha$ with $\alpha \in (0, \frac{1}{2})$). A property is said to be non-concentrated if only non-concentrated distributions can satisfy the property. 
We prove a lower bound on the testing of any non-concentrated property (not necessarily label-invariant).

\begin{theo}[Informal]\label{thm:nonconclb}
In order to non-tolerantly test any non-concentrated distribution property, $\Omega(\sqrt{n})$ samples are required, where $n$ is the size of the support of the distribution.
\end{theo}

The quadratic gap between tolerant testing and non-tolerant testing for any non-concentrated property follows from the above theorem, since by a folklore result, only $\Oh(n)$ many samples are required to learn any distribution approximately.

The proof of Theorem~\ref{thm:nonconclb} for label-invariant non-concentrated properties is a generalization of the proof of the $\Omega(\sqrt{n})$ lower bound for classical uniformity testing, while for the whole theorem, that is, for the general (not label-invariant) non-concentrated properties, a more delicate argument is required. The formal proof is presented in Section~\ref{sec:thm-noncon}.

The next natural question is about the sample complexity of any tolerant tester for non-concentrated properties.
We address this question for label-invariant non-concentrated properties by proving the following theorem in Section~\ref{sec:sym-tol}. However, the question is left open for non-label-invariant properties.

\begin{theo}[Informal]\label{thm:tol-sym-non}
The sample complexity for tolerantly testing any non-concentrated label-invariant distribution property is $\Omega(n^{1-o(1)})$, where $n$ is the size of the support of the distribution.
\end{theo}

A natural question related to tolerant testing is: 
\begin{center}
    \emph{How many samples are required to learn a distribution?}
\end{center}

As pointed out earlier, any distribution can be learnt using $\Oh(n)$ samples.
But what if the distribution happens to be \emph{very concentrated}?
We present an upper bound result for learning a distribution, in which the sample complexity depends on the {minimum} cardinality of any set $S \subseteq [n]$ over which the unknown distribution is concentrated.

\begin{theo}[Informal]\label{thm:nonconcub}
{To learn a distribution approximately, $\Oh(\size{S})$ samples are enough, where $S\subseteq [n]$ is an unknown set of minimum cardinality whose mass is close to $1$. Note that $\size{S}$ is also unknown, and the algorithm adapts to it. \remove{Note that trivially $S$ can be set to $[n]$. But $\size{S}$ can be much less than $n$ as we need the mass of $S$ and $1$ differ by at most a factor that depends on the proximity parameters.}}
\end{theo}

Observe that we cannot learn a distribution supported on the set $S$ using $o(|S|)$ samples, so the above result is essentially tight.

\subsection{Related works}\label{sec:rel}

Several forms of distribution testing have been investigated for over a century in statistical theory~\cite{king1997guide, corder2014nonparametric}, while combinatorial properties of distributions have been explored over the last two decades in Algorithm Theory, Machine Learning and Information Theory~\cite{goldreich2017introduction,MacKay2003,cover1999elements,bhattacharyya2022property}.
In Algorithm Theory, the investigation into testing properties of distributions started with the work of Goldreich and Ron~\cite{DBLP:journals/eccc/ECCC-TR00-020}, even though it was not directly stated there in these terms. Batu, Fortnow, Rubinfeld, Smith and White~\cite{Batu00} launched the intensive study of property testing of distributions with the problem of equivalence testing~\footnote{Given two unknown probability distributions that can be accessed via samples from their respective oracles, equivalence testing refers to the problem of distinguishing whether they are same or far from each other.}. Later, Batu, Fischer, Fortnow, Kumar, Rubinfeld and White~\cite{batu2001testing} studied the problems of identity and independence testing of distributions~\footnote{Given an unknown distribution accessible via samples, the problem of identity testing refers to the problem of distinguishing whether it is identical to a known distribution or far from it.}. Since then there has been a flurry of interesting works in this model. For example, Paninski~\cite{paninski2008coincidence} proved tight bounds on uniformity testing, Valiant and Valiant~\cite{ValiantV11} resolved the tolerant sample complexity for a large class of label-invariant properties that includes uniformity testing, Acharya, Daskalakis and Kamath~\cite{acharya2015optimal} proved various optimal testing results under several distance measures, and Valiant and Valiant~\cite{valiant2017automatic} studied the sample complexity of instance optimal identity testing. In \cite{batu2017generalized}, Batu and Cannone studied the problem of \emph{generalized uniformity testing}, where the distribution is promised to be supported on an unknown set $S$, and proved a tight bound of $\tilde{\Theta}(\size{S}^{2/3})$ samples for non-tolerant uniformity testing. This is in contrast to the non-tolerant uniformity testing of a distribution supported over $[n]$, whose sample complexity is $\Theta(\sqrt{n})$, ignoring the dependence on the proximity parameter.
Daskalakis, Kamath and Wright~\cite{daskalakis2018distribution} studied the problem of tolerant testing under various distance measures. Very recently,  Canonne, Jain, Kamath and Li~\cite{canonne2021price} revisited the problem of determining the sample complexity of tolerant identity testing, where they proved the optimal dependence on the proximity parameters. Going beyond studying specific properties, Canonne, Diakonikolas, Gouleakis and Rubinfeld~\cite{canonne2018testing} studied the class of \emph{shape-restricted} properties of a distribution, a condition general enough to contain several interesting properties like monotonicity, log-concavity, $t$-modality etc. Their result was later improved by Fischer, Lachish and Vasudev~\cite{fischer2019improving}.
See the surveys of Cannone~\cite{canonne2020survey, CanonneTopicsDT2022} for a more exhaustive list.

While the most studied works concentrate on non-tolerant testing of distributions, a natural extension is to test such properties tolerantly. Since the introduction of tolerant testing in the pioneering work of Parnas, Ron and Rubinfeld~\cite{parnas2006tolerant}, that defined this notion for classical (non-distribution) property testing, there have been several works in this framework. Note that it is nontrivial in many cases to construct tolerant testers from their non-tolerant counterparts, as in the case of tolerant junta testing~\cite{blais2019tolerant} for example. In a series of works, it has been proved that tolerant testing of the most natural distribution properties, like  uniformity, requires an almost linear number of samples~\cite{valiant2011testing, ValiantV11}~\footnote{To be precise, the exact lower bounds for non-tolerant uniformity testing is $\Omega(\sqrt{n})$, and for tolerant uniformity testing it is $\Omega(\frac{n}{\log n})$, where $n$ is the support size of the distribution and the proximity parameter $\eps$ is constant.}.
Now a natural question arises about how the sampling complexity of tolerant testing is related to non-tolerant testing of distributions in general. To the best of our knowledge, there is no known example with more than a quadratic gap.

It would also be interesting to bound the gap for sample-based testing as defined in the work of Goldreich and Ron~\cite{goldreich2016sample}. This model was investigated further in the work of  Fischer, Lachish and Vasudev~\cite{fischer2015trading}, where a general upper bound for non-tolerant sample-based testing of strongly testable properties was provided.

\paragraph*{Organization of the paper}
Section~\ref{sec:pre} contains the definitions used throughout the paper. The high-level overview of the proofs of all our theorems are presented in Section~\ref{sec:overview}. Section~\ref{sec:symresult} contains the formal statement and proof of Theorem~\ref{thm:symtolerant}. Theorem~\ref{thm:inf-conv} has been formally stated and proved in Section~\ref{sec:consttoltest}. Finally Theorem~\ref{thm:nonconclb}, Theorem~\ref{thm:tol-sym-non} and Theorem~\ref{thm:nonconcub} are formally stated and proved in Section~\ref{sec:nonconclb}, Section~\ref{sec:thm-noncon} and Section~\ref{sec:tol-ub} respectively.

\section{Notation and definitions}
\label{sec:pre}

A probability distribution $D$ over a universe $\Omega=[n]$ is a non-negative function $D:\Omega \rightarrow [0,1]$ such that $\sum_{i \in \Omega}D(i)=1$.
For $S \subseteq \Omega$, the mass of $S$ is defined as $D(S) = \sum_{i \in S}D(i)$, where $D(i)$ is the mass of $i$ in $D$. The support of a probability distribution $D$ on $\Omega$ is denoted by $\mbox{{\sc Supp}}(D)$. For any distribution $D$, by the top $t$ elements of $D$, we refer to the first $t$ elements in the support of $D$ when the elements in the support are sorted according to a non-increasing order of their probability masses in $D$. When we write $\tOh(\cdot)$, it sometimes suppresses a poly-logarithmic term in $n$ and the inverse of the proximity parameter(s), as well as the inverse of the difference of two proximity parameters. Although there are several other distance measures, in this work, we mainly focus on the $\ell_1$ distance.

\begin{defi}[Distribution property]
Let $\cD$ denote the set of all distributions over $\Omega$. A \emph{distribution property} $\cP$ is a topologically closed subset of $\cD$~\footnote{We put this restriction to avoid formalism issues. In particular, the investigated distribution properties that we know of (such as monotonicity and being a k-histogram) are topologically closed.}. A distribution $D \in \cP$ is said to be \emph{in the property} or to {\em satisfy} the property. Otherwise, $D$ is said to be \emph{not in the property} or \emph{to not satisfy} the property.
\end{defi}

\begin{defi}[Label-invariant property]
Let us consider a property $\cP$. For a distribution $D$ and a permutation $\sigma:\Omega \rightarrow \Omega$, consider the distribution $D_{\sigma}$ defined as $D_{\sigma}(\sigma(i))=D(i)$ (equivalently, $D_{\sigma}(i)=D(\sigma^{-1}(i)))$ for each $i \in \Omega$. If for every distribution $D$ in $\cP$, $D_\sigma$ is also in $\cP$ for every permutation $\sigma$, then the property $\cP$ is said to be \emph{label-invariant}.
\end{defi}

\begin{defi}[Distance between two distributions]
The distance between two distributions $D_1$ and $D_2$ over $\Omega$ is the standard $\ell_1$ distance between them, which is defined as $||D_1-D_2||_1:=\sum\limits_{i \in \Omega}\size{D_1(i)-D_2(i)}$. {For $\eta \in [0,2]$, $D_1$ and $D_2$ are said to be $\eta$-close to each other if $||D_1-D_2||_1 \leq \eta$. Similarly, for $\eps \in [0,2]$, $D_1$ and $D_2$ are said to be $\eps$-far from each other if $||D_1-D_2||_1 \geq \eps$}.
\end{defi}

\begin{defi}[Distance of a distribution from a property]\label{defi:distprop}
The distance of a distribution $D$ from a property $\cP$ is the minimum $\ell_1$-distance between $D$ and any distribution in $\cP$. For $\eta \in [0,2]$, a distribution $D$ is said to be $\eta$-close to $\cP$ if the distance of  $D$ from $\cP$ is at most $\eta$. Analogously, for $\eps \in [0,2]$, a distribution $D$ is said to be $\eps$-far from $\cP$ if the distance of $D$ from $\cP$ is at least $\eps$.
\end{defi}

\begin{defi}[$(\eta, \eps)$-tester]
An \emph{$(\eta, \eps)$-tester} for a distribution property is a randomized algorithm that has sample access to the unknown distribution (upon query it can receive elements of $\Omega$, each drawn according to the unknown distribution, independently of any previous query or the algorithm's private coins), and distinguishes whether the distribution is $\eta$-close to the property or $\eps$-far from the property, with probability at least $\frac{2}{3}$, where $\eta$ and $\eps$ are proximity parameters such that $0 \leq \eta < \eps \leq 2$. The tester is said to be \emph{tolerant} when $\eta > 0$, and \emph{non-tolerant} when $\eta = 0$.
\end{defi}

Now we define the notions of non-concentrated distributions and non-concentrated properties.

\begin{defi}[Non-Concentrated distribution]\label{defi:non-conc-dist}
A distribution $D$ over the domain $\Omega=[n]$ is said to be \emph{$(\alpha, \beta)$-non-concentrated} if for any set $S\subseteq \Omega$ with size $\beta n$, the probability mass on $S$ is at least $\alpha$, where $\alpha$ and $\beta$ are two parameters such that $0 < \alpha \leq \beta < \frac{1}{2}$.
\end{defi}

\begin{defi}[Non-Concentrated property]\label{defi:non-conc-prop} Let $0 < \alpha \leq \beta < \frac{1}{2}$. 
A distribution property $\mathcal{P}$ is defined to be \emph{$(\alpha, \beta)$-non-concentrated}, if all distributions in $\cP$ are $(\alpha, \beta)$-non-concentrated. 
\end{defi}

Note that the uniform distribution is $(\alpha,\alpha)$-non-concentrated for every $\alpha$, and so is the property of being identical to the uniform distribution. Also, for any $0< \alpha < \frac{1}{2}$ such that $\alpha n$ is  an integer, the uniform distribution is the only $(\alpha,\alpha)$-non-concentrated one. Finally, observe that any arbitrary distribution is both $(0,\beta)$-non-concentrated and $(\alpha,1)$-non-concentrated, for any $\alpha, \beta \in (0, 1)$.

\subsection{Useful concentration bounds}\label{sec:prelim}

Now we state some large deviation results that will be used throughout the paper.

\begin{lem}[{\bf Multiplicative Chernoff bound}~\cite{DubhashiP09}]
\label{lem:cher_bound1}
Let $X_1, \ldots, X_n$ be independent random variables such that $X_i \in [0,1]$. For $X=\sum\limits_{i=1}^n X_i$ and $\mu=\E[X]$, the following holds for any $0\leq \delta \leq 1$.
$$ \pr(\size{X-\mu} \geq \delta\mu) \leq 2\exp{\left(-\frac{\mu \delta^2}{3}\right)}$$
\end{lem}

\begin{lem}[{\bf Additive Chernoff bound}~\cite{DubhashiP09}]
\label{lem:cher_bound2}
Let $X_1, \ldots, X_n$ be independent random variables such that $X_i \in [0,1]$. For $X=\sum\limits_{i=1}^n X_i$ and $\mu_l \leq \E[X] \leq \mu_h$, the following hold for any $\delta >0$.
\begin{itemize}
\item[(i)] $\pr \left( X \geq \mu_h + \delta \right) \leq \exp{\left(-\frac{2\delta^2}{n}\right)}$.
\item[(ii)] $\pr \left( X \leq \mu_l - \delta \right) \leq \exp{\left(-\frac{2\delta^2} { n}\right)}$.
\end{itemize}

\end{lem}

\section{Overview of our results}\label{sec:overview}

In this section, we give an overview of our results as follows:

\subsection{Overview of proof of Theorem~\ref{thm:symtolerant} and Theorem~\ref{thm:inf-conv}}

We first show that for any label-invariant distribution property, the sample complexities of tolerant and non-tolerant testing are separated by at most a quadratic factor, ignoring poly-logarithmic terms. More specifically, in Theorem~\ref{thm:symtolerant}, we prove that for any label-invariant distribution property $\cP$ that has a non-tolerant tester with sample complexity $\Lambda$, there exists a tolerant tester for $\cP$ that uses $\widetilde{\Oh}(\Lambda^2)$ samples, ignoring poly-logarithmic factors. Since we can learn a distribution using $\Oh(n)$ samples, our proof is particularly useful when $\Lambda = o(\sqrt{n})$, where $n$ is the  size of the support of the distribution that is being tested.

To prove Theorem~\ref{thm:symtolerant} (restated as Theorem~\ref{thm:tvsnt}), we provide an algorithm for tolerant testing of $\cP$ with sample complexity $\tOh(\Lambda^2)$, based on the existence of a non-tolerant tester of $\cP$ with sample complexity $\Oh(\Lambda)$. Given the existence of such a non-tolerant tester with sample complexity $\Oh(\Lambda)$, one crucial observation that we use here is that there cannot be two distributions $D_1$ and $D_2$ that are identical on the elements with mass $\Omega(\frac{1}{\Lambda^2})$ (we call them \emph{high} elements), where $D_1$ is in the property $\cP$ while $D_2$ is far from $\cP$. This is formally stated as Lemma~\ref{lem:lbnon}. 

Given that the two distributions $D_1$ and $D_2$ are identical on all elements with mass  $\Omega(\frac{1}{\Lambda^2})$, by the birthday paradox, we can say that $\Oh(\Lambda)$ samples are not enough to obtain any \emph{low} elements, that is, elements with mass $o(\frac{1}{\Lambda^2})$, that appear more than once. Since the property $\cP$ is label-invariant, we can apply uniformly random permutations over the low elements of both $D_1$ and $D_2$, making the samples obtained from both $D_1$ and $D_2$ appear as two uniformly random sequences. Thus, from the view of any tester that takes only $\Oh(\Lambda)$ samples, $D_1$ and $D_2$ will appear the same, which would contradict the existence of a non-tolerant tester that distinguishes $D_1$ from $D_2$ using $\Oh(\Lambda)$ samples. At this point, we would like to point out that the proof of Lemma~\ref{lem:lbnon} only assumes the existence of a non-tolerant tester, and is oblivious to its internal details. Later, in Lemma~\ref{cor:lbnom}, we generalize this idea to show that when $D_1$ and $D_2$ are close with respect to the high elements, it cannot be the case that $D_1$ is in the property $\cP$, while $D_2$ is far from $\cP$. Although the proof follows a similar line to that of Lemma~\ref{lem:lbnon}, more careful analysis is required to prove Lemma~\ref{cor:lbnom}. Note that Lemma~\ref{cor:lbnom} is the main technical lemma required to prove Theorem~\ref{thm:symtolerant}.

Once we have Lemma~\ref{lem:lbnon} and Lemma~\ref{cor:lbnom}, we can describe the algorithm of Theorem~\ref{thm:symtolerant}. Broadly speaking, we show that \emph{partial learning} of the distribution is sufficient for constructing a tolerant tester for any label-invariant property, as opposed to the more familiar paradigm of \emph{testing by learning}~\cite{diakonikolas2007testing, servedio2010testing}. Using Lemma~\ref{cor:lbnom}, we show that estimating the masses of only the high elements is enough for us, along with the fact that the property $\cP$ that we are testing is label-invariant. Roughly, the algorithm has three steps. In the first step, we identify and measure the high elements of the unknown distribution $D$. In the second step, we construct a new distribution $\widetilde{D}$ that adheres to the high mass elements obtained from the first step. Finally, in the third step, we check whether there exists any distribution $D_1$ in $\cP$ that is close to $\widetilde{D}$. If such a distribution exists, we accept, and otherwise we reject. In the first step, we need $\tOh(\Lambda^2)$ many samples to correctly estimate the masses of the high elements, which dominates the sample complexity of our tolerant tester.

It is important to note that the \emph{computational efficiency} of the tolerant tester depends on how fast we can check whether the distribution $\widetilde{D}$ (constructed by the algorithm) is {\em close} to a known property $\cP$, where we have the complete description of $\widetilde{D}$. Later, in Theorem~\ref{thm:inf-conv} (restated as Theorem~\ref{theo:tol-conv}), we show that when the property $\cP$ can be expressed as a feasible solution to a set of linear inequalities, there exists an algorithm that tolerantly tests for $\cP$ in time polynomial in the support size of the distribution and the number of linear inequalities required to represent it. The algorithm is similar to that of Theorem~\ref{thm:symtolerant}, whereas its polynomial running time follows by using the Ellipsoid method.

\subsection{Overview of Theorem~\ref{thm:nonconclb} and Theorem~\ref{thm:tol-sym-non}}
In Theorem~\ref{thm:nonconclb}, we show that in order to non-tolerantly test any non-concentrated property (defined in Definition~\ref{defi:non-conc-prop}), $\Omega(\sqrt{n})$ samples are required, where $n$ denotes the support size of the distribution. Before directly proceeding to prove the result, as a warm-up, we first show an analogous result for label-invariant non-concentrated properties in Theorem~\ref{theo:sym}. To prove the theorem, for any distribution $D_{yes}$ in the label-invariant non-concentrated property $\cP$ that we are testing, we construct a new distribution $D_{no}$ that is far from $\cP$, whose support is a subset of the support of $D_{yes}$. The two distributions are identical over their high probability elements, and they only differ in their \emph{low} probability elements, where a low probability element is an element with mass $\Oh(\frac{1}{n})$. Since $D_{yes}$ and $D_{no}$ differ only on the elements with mass $\Oh(\frac{1}{n})$, by the birthday paradox and the fact that the property is label-invariant, any tester that takes $o(\sqrt{n})$ samples cannot distinguish between $D_{yes}$ and $D_{no}$, and the result follows. We note that the proof of Theorem~\ref{theo:sym} is a generalization of the lower bound proof for uniformity testing.

Though the proof of Theorem~\ref{thm:nonconclb} (restated as Theorem~\ref{thm:non-sym}) follows similarly to that of Theorem~\ref{theo:sym}, delicate analysis is required to take care of the fact that the properties are no longer label-invariant. We also discuss briefly the reason why the technique used to prove Theorem~\ref{theo:sym} does not work to prove Theorem~\ref{thm:nonconclb}, in the beginning of Section~\ref{sec:thm-noncon}.

As a step further, in Theorem~\ref{thm:tol-sym-non} (restated as Theorem~\ref{theo:tol-sym-lb}), we show that $\Omega(n^{1- o(1)})$ samples are necessary to tolerantly test any non-concentrated label-invariant property. This proof follows from an application of the \emph{low frequency blindness theorem} of Valiant~\cite{valiant2011testing}. The question of tolerant testing of general non-concentrated properties remains open.

\subsection{Overview of Theorem~\ref{thm:nonconcub}}

Finally, we consider the problem of learning a distribution $D$, where $D$ is \emph{concentrated} over a unknown set $S \subseteq \Omega$. In Theorem~\ref{thm:nonconcub} (restated as Theorem~\ref{theo:learn-main}), we give an algorithm that achieves this with $\Oh(\size{S})$ many samples, even when $|S|$ is also unknown. Note that this problem is reminiscent of the folklore result of learning a distribution over any set $S$ that takes $\Oh(\size{S})$ samples. However, the folklore result holds only for the case where the set $S$ is known~\footnote{There are also prior results where only $|S|$ is known, such as in the work of Acharya, Diakonikolas, Li and Schmidt~\cite{acharya2017sample}.}.

Broadly, the algorithm iterates over possible values of $|S|$. Starting from $s=1$, we first take $s$ many samples from the the unknown distribution $D$, and construct a new empirical distribution $D_s$ based upon the samples obtained. Once we have the distribution $D_s$, we apply the result of Valiant and Valiant~\cite{ValiantV11} to test whether the unknown distribution $D$ is close to the newly constructed distribution $D_s$, by using number of samples that is slightly smaller than $s$. If $D_s$ is close to $D$, we report the distribution $D_s$ as the output and terminate the algorithm. Otherwise, we double the value of $s$ and perform another iteration of the two steps as mentioned above. Finally, we show that when $s \geq \size{S}$, where $S$ is the unknown set on which $D$ is concentrated, $D_s$ will be close to $D$ with high probability, and we will output a distribution satisfying the statement of Theorem~\ref{thm:nonconcub}. To the best of our knowledge, this is the first result of a tester of this kind that adapts to an unknown support size $|S|$.

\section{Non-tolerant vs.\ tolerant sample complexities of label-invariant properties (Proof of Theorem~\ref{thm:symtolerant})}\label{sec:symresult}

We will prove that for any label-invariant property, the sample complexities of tolerant and non-tolerant testing are separated by at most a quadratic factor (ignoring some poly-logarithmic factors). Formally, the result is stated as follows:

\begin{theo}[Theorem~\ref{thm:symtolerant} formalized]\label{thm:tvsnt}
Let $\mathcal{P}$ be a label-invariant distribution property, for which there exists an  $(0, \eps)$-tester (non-tolerant tester) with sample complexity $\Lambda(n,\eps)$,  where $\Lambda \in \N$ and $0 < \eps \leq 2$. Then for any $ \gamma_1, \gamma_2$ with $\gamma_1<\gamma_2$ and $0<\gamma_2 + \eps < 2$, there exists a $(\gamma_1, \gamma_2 + \eps)$-tester (tolerant tester) that has sample complexity
$\Oh\left(\frac{1}{(\gamma_2 - \gamma_1)^3} \cdot \min\{{\Lambda^2 \log^2 \Lambda}, n\}\right)$, where $\Lambda=\Lambda(n,\epsilon)$, and $n$ is the size of the support of the distribution.
\end{theo}

Let us assume that $D$ is the unknown distribution and $\Lambda(n,\epsilon) \geq \Omega(\frac{1}{\eps})$~\footnote{This is a reasonable assumption for any non-trivial property.}. First note that if $\Lambda=\Omega{(\sqrt{n})}$, then we can construct a distribution $\widehat{D}$ such that $||D-\widehat{D}||_1 < \frac{\gamma_2 - \gamma_1 + \eps}{2}$, by using $\Oh\left(\frac{n }{(\gamma_2 - \gamma_1 + \eps)^2}\right)$ samples from $D$. Thereafter we can report $D$ to be $\gamma_1$-close to the property if and only if $\widehat{D}$ is $\frac{\gamma_2+ \gamma_1 + \eps}{2}$-close to the property. In what follows, we discuss an algorithm with sample complexity $\tOh(\Lambda^2)$ when $\Lambda=o(\sqrt{n})$. Also, we assume that $n$ and $\Lambda$ are larger than some suitable constant. Otherwise, the theorem becomes trivial.

The idea behind the proof is to classify the elements of $\Omega$ with respect to their masses in $D$ into \emph{high} and \emph{low}, as formally defined below in Definition~\ref{defi:highlow}. We argue that since $\mathcal{P}$ is $(0, \eps)$-testable using  $\Lambda(n,\eps) = \Oh(q)$ samples, there cannot be two distributions $D_1$ and $D_2$ that are identical on all elements whose probability mass is at least $\frac{1}{q^2}$, for $q=\theta(\Lambda)$ (the set $\high_{{1}/{q^2}}$ defined below), where $D_1 \in \mathcal{P}$ but $D_2$ is $\eps$-far from $\mathcal{P}$. We will formally show this in Lemma~\ref{lem:lbnon}, where we will use the fact that $\cP$ is label-invariant. Using Lemma~\ref{lem:lbnon}, we prove Lemma~\ref{cor:lbnom}, that (informally) says that if two distributions are close with respect to the high mass elements, then it is not possible that one distribution is close to $\cP$ while the other one is far from it.  In our algorithm, we intend to approximate the masses of the set
$\high_{{1}/{q^2}}$, and the term $\Lambda^2$ in the query complexity of our algorithm corresponds to that.

\begin{defi}\label{defi:highlow}
For a distribution $D$ over $\Omega$ and $0< \kappa < 1$, we define 
$$\high_{\kappa}(D) = \{x\in \Omega\ \mid\ D(x)\geq \kappa  \} 
$$
\end{defi}

Now we define a quantity $q \in \N$ where $q=\Theta(\Lambda)$~\footnote{Note that $q$ and $\Lambda$ are of the same order of magnitude. We have introduced $q$ for writing proofs more rigorously.}. Assume that $c^*$ is a suitable large constant (independent of $\Lambda$) such that, if we take $\Lambda$ many samples from a distribution, then with probability at least $\frac{3}{4}$, we will not get any sample $x$ whose mass is at most $(\frac{c^*}{\Lambda})^2$ more than once. We define 
\begin{equation}\label{eqn:lambda}
q:= \frac{\Lambda}{c^*}.
\end{equation}


We will complete the proof of Theorem~\ref{thm:tvsnt} by using the following two lemmas which we will prove later.

\begin{lem}\label{lem:lbnon}
Let $\mathcal{P}$ be  a label-invariant property that is $(0, \eps)$-testable using  $\Lambda(n,\eps)$ samples and consider $q$ as defined in Equation~\ref{eqn:lambda}. Let $D_1$ and $D_2$ be two distributions such that $\high_{{1}/{q^2}}(D_1) = \high_{{1}/{q^2}}(D_2)$, and for all $x\in \high_{{1}/{q^2}}(D_1)$, the probability of $x$ is the same for both distributions, that is, $D_1(x) = D_2(x)$. Then it is not possible that $D_1$ satisfies $\mathcal{P}$ while $D_2$ is $\eps$-far from satisfying $\mathcal{P}$.
\end{lem}


\begin{lem}\label{cor:lbnom}
Let $\mathcal{P}$ be a label-invariant property that is $(0, \eps)$-testable using  $\Lambda(n,\eps)$ samples,  and consider $q$ as defined in Equation~\eqref{eqn:lambda}. Let $\overline{D}$ and $\widetilde{D}$ be two distributions over $\Omega$, where $\size{\Omega} > 4 q^2$, and let $H$ contain the top $q^2$ elements of $\overline{D}$. Also, assume that $\size{\widetilde{D}(\Omega\setminus H) -\overline{D}(\Omega\setminus H)} \leq \gamma$.
If \begin{equation}\label{eq:corlbnon}
    \sum\limits_{x\in H}\left|\overline{D}(x) - \widetilde{D}(x) \right| \leq \alpha,
    \end{equation} 
then the following hold:

\begin{enumerate}
    \item   If $\overline{D}$ is $\beta$-close to $\mathcal{P}$, then there exists a distribution $D_1$ in $\mathcal{P}$ such that  $\high_{1/q^2}(D_1) \subseteq H$ and
    \begin{equation}\label{eq:corlbnon1}
    \sum\limits_{ x\in H}\left|D_1(x) - \widetilde{D}(x) \right|  + \left|{D_1}(\Omega \setminus H) - \widetilde{D}(\Omega \setminus H) \right|\leq (\alpha + \beta + \gamma).
    \end{equation} 

    \item If $\overline{D}$ is $(\eps + 3\alpha + \beta + 2\gamma)$-far from $\mathcal{P}$ and $D_1$ is a distribution such that $\high_{1/q^2}(D_1) \subseteq H$ and
    \begin{equation}\label{eq:corlbnon2}
    \sum\limits_{x\in H}\left|D_1(x) - \widetilde{D}(x) \right| +  \left|{D_1}(\Omega \setminus H) - \widetilde{D}(\Omega \setminus H) \right|\leq (\alpha + \beta + \gamma),
    \end{equation} 
    then the distribution $D_1$ does not satisfy the property $\mathcal{P}$.
\end{enumerate}

\end{lem}

Using the above two lemmas,  we will prove Theorem~\ref{thm:tvsnt} in Section~\ref{sec:thm-3.1}.
We present
the proofs of Lemma~\ref{lem:lbnon} and Lemma~\ref{cor:lbnom} in Section~\ref{subsec:lbnon}.

\subsection{Proof of Theorem~\ref{thm:tvsnt}}\label{sec:thm-sym}\label{sec:thm-3.1}
Let $D$ be the unknown distribution that we need to test, and assume that  $\zeta = \gamma_1$, $\eta = {\gamma_2 - \gamma_1}$, and $\eta'=\frac{\eta}{64}$.
We now provide a tolerant 
$(\gamma_1, \gamma_2 + \eps)$-tester, that is,  a $(\zeta, \zeta + \eps + \eta)$-tester for the property $\mathcal{P}$, as follows:

\begin{enumerate}

    \item Draw $W = \Oh\left(\frac{q^2}{\eta'}\log q \right)$ many samples from the distribution $D$. Let $S\subseteq \Omega$ be the set of (distinct) samples obtained.
    
    \item  Draw additional $\Oh \left(\frac{W}{\eta'^2}\log W \right)$ many samples $Z$ to estimate the value of $D(x)$ for all $x\in S$~\footnote{Instead of two sets of random samples (where the first one is to generate the set $S$ and the other one is the multi-set $Z$), one can work with only one set of random samples. But in that case, the sample complexity becomes $\Oh(q^2\log n)$, as opposed to $\Oh(q^2 \log q)$ that we are going to prove.}.
    
     \item Construct a set $H$ as the union of $S$ and arbitrary $q^2$ many elements from $\Omega \setminus (S \cup Z)$. 
    
    \item Define a distribution $\widetilde{D}$ such that, for $x \in H$, 
    $$\widetilde{D}(x)=\frac{\#~ x~\mbox{in the multi-set}~ Z}{|Z|}.$$
    
   And for each $x\in \Omega\setminus H$, 
    $$\widetilde{D}(x) =\frac{1-\sum\limits_{x \in H}\widetilde{D}(x)}{\size{\Omega}-\size{H}}.$$

    \item
    
    If there exists a distribution $D_1$ in $\mathcal{P}$ that satisfies both the following conditions:
    
    \begin{enumerate}
        \item[\textbf{(A)}] $\sum\limits_{ x \in H}\left|D_1(x) - \widetilde{D}(x) \right|  + |D_1(\Omega\setminus H) - \widetilde{D}(\Omega\setminus H)| \leq  26 \eta'+ \zeta $.
        \item[\textbf{(B)}] $\high_{1/q^2}(D_1) \subseteq H$\label{eq:corlbnon4}.
    \end{enumerate}
    
    then ACCEPT $D$.

    \color{black}

    \item If there does not exist any $D_1$ in $\mathcal{P}$ that satisfies both Conditions \textbf{(A)} and \textbf{(B)} above, then REJECT $D$. 
    
    \color{black}

\end{enumerate}

Note that Step $5$ as mentioned above is not completely constructive in a computational sense. In Section~\ref{sec:consttoltest}, we give a constructive variant of the tester where the property $\cP$ can be expressed as a set of linear inequalities. We also give an example of a natural property that can be expressed as a set of linear inequalities.

\paragraph*{Sample Complexity.} The sample complexity of the tester is  $\Oh(\frac{q^2}{\eta^3}\log^2 q) = \Oh(\frac{{\Lambda^2 \log^2 \Lambda}}{(\gamma_2 - \gamma_1)^3})$, which follows from the above description.

\paragraph*{Correctness of the algorithm.}

The proof of correctness of our algorithm is divided into a sequence of lemmas.

\begin{lem}\label{lem:cher} The set $H$ and the distribution $\widetilde{D}$ satisfies the following three properties:
\begin{description}
\item[(i)]  With probability at least $1-\frac{1}{q}$, $\high_{\eta'/q^2}(D)\subseteq S \subseteq H$. 
\item[(ii)] For any $x \in H$, if $D(x) \geq \frac{\eta'}{10W}$,
 $(1 - \eta')D(x) \leq \widetilde{D}(x) \leq (1 + \eta')D(x)$ holds with probability at least $1- \frac{1}{q^4}$.

\item[(iii)] For any $x \in \Omega$ with $D(x)\leq \frac{\eta'}{10W}$, either $x \notin H$, or $\widetilde{D}(x)\leq (1+\eta')\frac{\eta'}{10W}$ holds with probability at least $1- \frac{1}{q^4}$.
\end{description} 
\end{lem}


\begin{proof} Let us prove the three parts one by one:
\begin{itemize}
    \item \textbf{(i)} Consider any $x \in \high_{\eta'/q^2}(D)$, that is, $D(x)\geq \frac{\eta'}{q^2}$. Then the probability that $x \notin H$ is at most $(1- \frac{\eta'}{q^2})^{\size{H}}\leq \frac{1}{q^4}$. Applying the union bound over all the elements in $\high_{\eta'/q^2}(D)$ (at most $\frac{q^2}{\eta'}= \Oh(q^3)$~\footnote{This follows from the assumption that $\Lambda(n,\epsilon)$ is at least $\Omega(1/\epsilon)$.} many elements), the claim follows.
    
    \item \textbf{(ii)} Since $\size{Z} = \Oh(\frac{W}{\eta'^2} \log W)$, applying Chernoff bound, we see that  $(1 - \eta')D(x) \leq \widetilde{D}(x) \leq (1 + \eta')D(x)$ does not hold with probability at most $\frac{1}{q^4}$.
    
    \item \textbf{(iii)} 
    Since $\size{Z} = \Oh(\frac{W}{\eta'^2}\log W)$, if $x$ is in $H$ (otherwise, we are already done), applying Chernoff bound (only on one side), the bound follows.

\end{itemize}
\end{proof}

We now bound the $\ell_1$-distance between $D$ and $\widetilde{D}$ with respect to $H$. 

\begin{lem}\label{lem:lem3}
$\sum\limits_{x \in H}\size{{D}(x)-\widetilde{D}(x)} \leq 5\eta'(1 + \eta') \leq 10\eta'$ holds with probability at least $1-\frac{3}{q}$. 
\end{lem}

\begin{proof}
 Recall the definition of $\high_{\eta'/10W}({D})$. Note that 
\begin{align*}
    \sum\limits_{x \in H}\size{D(x)-\widetilde{D}(x)} 
    &=\sum\limits_{x \in  \high_{\eta'/10W}({D}) }\size{D(x)-\widetilde{D}(x)} + \sum\limits_{x \in  H \setminus \high_{\eta'/10W}({D}) }\size{D(x)-\widetilde{D}(x)} 
\end{align*}

Applying Lemma~\ref{lem:cher} $(ii)$ for each $x \in \high_{\eta'/10W}({D})$, and then using union bound over all such $x\in \high_{\eta'/10W}({D})$, the first term is bounded by $\eta'$ with probability at least $1-\frac{1}{q}$. 

Now the second term, notice that for each $x \in H \setminus  \high_{\eta'/10W}({D})$, $D(x)\leq \frac{\eta'}{10W}$. 
By Lemma~\ref{lem:cher} (iii), and using the union bound over all elements in $H \setminus \high_{\eta'/10W}({D})$ (note that $|H| \leq 2W=\Oh(q^3))$, with probability at least $1-\frac{2}{q}$, $\widetilde{D}(x) \leq \eta'(1 + \eta')/10W$ for all $x \in H \setminus \high_{\eta'/10W}({D})$. Since $|H|\leq 2W$, the second term is bounded by $4\eta'(1 + \eta')$ with probability at least $1-\frac{2}{q}$.
\end{proof}

Now we prove a lemma that shows that for every distribution $D$, there is a another distribution $\overline{D}$ that is ``similar'' to $D$, and for which $H$ contains the top $q^2$ elements of $\overline{D}$.

\color{black}

\begin{lem}\label{lem:lem4}
There exists a distribution $\overline{D}$ such that $H$ contains the top $q^2$ elements of $\overline{D}$. Moreover, the following hold:
\begin{itemize}
\item[(i)]    $||D-\overline{D}||_1 \leq 2\eta'$, with probability at least $1-\frac{2}{q}$. 
    \item[(ii)] $\sum\limits_{x \in H}\size{\overline{D}(x)-\widetilde{D}(x)} \leq 12\eta'$, with probability at least $1-\frac{5}{q}$. 
    \item[(iii)] $|\overline{D}(\Omega \setminus H) - \widetilde{D}(\Omega \setminus H)| \leq 12\eta'$, with probability at least $1-\frac{5}{q}$. 
\end{itemize}

\end{lem}

\begin{proof}

Let $T$ be the set of $q^2$ largest elements of $D$. If $T \subseteq S$,  $H$ (as $S \subset H$) contains the largest $q^2$ elements of $D$. In that case, setting $\overline{D}$ to be $D$ gives us the above results.


Now, let us consider the case where $T \not\subseteq S$. By Lemma~\ref{lem:cher} (part (i)), with probability at least $1-\frac{2}{q}$, $\high_{\eta'/q^2}(D)\subseteq S$. Thus for any $x \in H\setminus S$, $D(x) < \frac{\eta'}{q^2}$. Consider the set $U = T\setminus H$. Notice that since $|H\setminus S|=q^2$ and $|T|=q^2$, $|U| \leq |H\setminus (T\cup S)|$. Let $U=\{y_1, \dots, y_{|U|}\}\subset \Omega \setminus H$, and let $z_1. \dots, z_{|U|}$ be some $|U|$ elements of $H\setminus (T\cup S)$. Note, by definition of $T$ and $U$, the set $\{z_1. \dots, z_{|U|}\}$ and the set $\{y_1. \dots, y_{|U|}\}$ are disjoint.

Consider the distribution $\overline{D}$ defined as follows:
\begin{itemize}
    \item For elements in $\{z_1. \dots, z_{|U|}\}$, we define $\overline{D}(z_i) = D(y_i)$.
    \item  For elements in $\{y_1. \dots, y_{|U|}\}$, we define $\overline{D}(y_i) = D(z_i)$.
    \item For all other $x$, we define $\overline{D}(x) = D(x)$.
\end{itemize}

Note that since all the elements in the sets $\{z_1. \dots, z_{|U|}\}$
and $\{y_1. \dots, y_{|U|}\}$ were from $\Omega\setminus S$, from
Lemma~\ref{lem:cher} (part (i)), with probability at least $1- \frac{2}{q}$,
$D(y_i) \leq \frac{\eta'}{q^2}$ and $D(z_i) \leq \frac{\eta'}{q^2}$, for all $i\in \{1, \ldots , \size{U}\}$. Moreover, as $|U|\leq q^2$, we have condition (i) as well. Furthermore, $H$ contains the largest $q^2$ elements of $\overline{D}$ due to its construction.

Using the triangle inequality (relative to $H$) along with Lemma~\ref{lem:lem3} and the above expression, we can say that, with probability at least $1 - \frac{5}{q}$, (ii) follows.


Let us now prove (iii). Since $\overline{D}$ and $\widetilde{D}$ are distributions,
$\sum\limits_{x \in H}\overline{D}(x)+ \sum\limits_{x \in \Omega \setminus H}\overline{D}(x)= \sum\limits_{x \in H}\widetilde{D}(x)+ \sum\limits_{x \in \Omega \setminus H} \widetilde{D}(x).$
Thus,
\begin{eqnarray*}
\size{\overline{D}(\Omega \setminus H) - \widetilde{D}(\Omega \setminus H)} &=& \size{\sum_{x \in H}\widetilde{D}(x) - \sum_{x \in H}\overline{D}(x)} \leq \sum_{x \in H}\size{\widetilde{D}(x) - \overline{D}(x)} \leq 12 \eta'
\end{eqnarray*}
The last inequality follows from (ii).
\end{proof}

Now we finally establish the correctness of the algorithm. 


\begin{proof}[Proof of correctness of the algorithm]
For completeness, consider the case where $D$ is $\zeta$-close to $\mathcal{P}$. By Lemma~\ref{lem:lem4} (i) and the triangle inequality, we know that there exists a distribution $\overline{D}$ that is $(\zeta + 2 \eta')$-close to $\cP$ and $H$ contain the largest $q^2$ elements of $\overline{D}$. 
Since $\sum\limits_{x \in H}\size{\overline{D}(x)-\widetilde{D}(x)} \leq 12\eta'$ and $\size{\overline{D}(\Omega \setminus H)-\widetilde{D}( \Omega \setminus H)} \leq 12\eta'$ hold from Lemma~\ref{lem:lem4} (ii) and (iii), following Lemma~\ref{cor:lbnom} for $ \alpha = 12 \eta'$, $\beta= \zeta + 2 \eta'$ and $\gamma= 12 \eta'$, we can say that there exists a distribution $D_1$ in $\mathcal{P}$ satisfying Equation~\eqref{eq:corlbnon1} (which is same as satisfying \textbf{Condition (A)} and \textbf{Condition (B)} in Step $5$ of the algorithm). Hence, our algorithm accepts $D$ in Step $5$.

For soundness, consider a distribution $D$ that is $(\eps + \zeta + \eta)$-far from $\cP$. Then following Lemma~\ref{lem:lem4} (i), we know that there exists a distribution $\overline{D}$ that is $(\eps + \zeta + \eta- 2\eta')$-far from $\cP$, that is, $(\eps+3 \alpha+ \beta + 2\gamma)$-far from $\cP$, where $\alpha=12\eta'$, $\beta=\zeta + 2 \eta'$. Here, we are using that $\eta=64\eta'$ and $\gamma = 12 \eta'$. Also Lemma~\ref{lem:lem4} guarantees that $H$ contains the top $q^2$ elements of $\overline{D}$.
Following Lemma~\ref{cor:lbnom}, we know that there does not exist any such distribution $D_1$ in $\mathcal{P}$ that satisfies both \textbf{Condition (A)} and \textbf{Condition (B)} of Step $5$ of the algorithm. 
Thus the algorithm will REJECT the distribution $D$ in Step $6$.

Note that the total failure probability of the algorithm is bounded by the probability that Lemma~\ref{lem:lem4} does not hold, which is at most $\frac{12}{q}$.
\end{proof}


\subsection{Proof of Lemma~\ref{lem:lbnon} and Lemma~\ref{cor:lbnom}}\label{subsec:lbnon}

\begin{proof}[Proof of Lemma~\ref{lem:lbnon}]
We will prove this by contradiction. Let us assume that there are two distributions $D_{yes}$ and $D_{no}$ such that 
 \begin{itemize}
     \item $D_{yes} \in \mathcal{P}$;
     
     \item $D_{no}$ is $\eps$-far from $\mathcal{P}$;

     \item $\high_{1/q^2}(D_{yes}) = \high_{1/q^2}(D_{no})=A$;
     
     \item For all $x\in A$, $D_{yes}(x) = D_{no}(x)$.
     
     \color{black}
     
 \end{itemize}

Now, we argue that any $(0,\eps)$-non-tolerant tester requires more than $\Lambda(n,\eps)$ samples from the unknown distribution $D$ to distinguish whether $D$ is in the property or $\eps$-far from it.
 
Let $D_Y$ be a distribution obtained from $D_{yes}$ by permuting the labels of $\Omega \setminus A$ using a uniformly random permutation. Specifically, consider a random permutation $\pi:\Omega \setminus A \rightarrow \Omega \setminus A$. The distribution $D_Y$ is as follows:
 \begin{itemize}
     \item  $D_Y(x)=D_{yes}(x)$ for each $x \in A$ and 
     \item $D_{Y}(\pi(x))=D_{yes}(x)$ for each $x \in \Omega \setminus A$.
 \end{itemize}
 
Similarly, consider the distribution $D_N$ obtained from $D_{no}$ by permuting the labels of $\Omega \setminus A$ using a uniformly random permutation. Note that $D_Y$ is in $\cP$, whereas $D_N$ is $\eps$-far from $\mathcal{P}$, which follows from $\cP$ being label-invariant. 

We will now prove that $D_{Y}$ and $D_{N}$ provide similar distributions over sample sequences. More formally, we will prove that any algorithm that takes at most  $\Lambda(n,\eps)$ many samples, cannot distinguish $D_{Y}$ from $D_{N}$ with probability at least $\frac{2}{3}$. We argue that this claim holds even if the algorithm is provided with additional information about the input: Namely, for all $x\in A$, the algorithm is told the value of $D_Y(x)$ (which is the same as $D_N(x)$). When the algorithm is provided with this information, it can ignore all samples obtained from ${A}$. 

By the definition of $A$, for all $x\in \Omega \setminus A$, both $D_Y(x)$ and $D_N(x)$ are at most $\frac{1}{q^2}$.
Let $S_Y$ be a sequence of samples drawn according to $D_Y$. If $|S_Y| \leq \Lambda(n,\eps)$, then with probability at least $\frac{3}{4}$, the sequence $(\Omega \setminus A)\cap S_Y$ has no element that appears twice. In other words, the set $(\Omega \setminus A)\cap S_Y$ is a set of at most $\Lambda(n,\eps)$ distinct elements from $\Omega \setminus A$. Since the elements of $\Omega \setminus A$ were permuted using a uniformly random permutation, with probability at least $\frac{3}{4}$, the sequence $( \Omega \setminus A)\cap S_Y$ is a uniformly random sequence of distinct elements from $\Omega \setminus A$.
Similarly, if $S_N$ is a sequence of samples drawn according to $D_N$, then with probability at least $\frac{3}{4}$, the sequence $(\Omega \setminus A)\cap S_N$ is a uniformly random sequence of distinct elements from $\Omega \setminus A$. Thus, the distributions over the received sample sequence obtained from $D_Y$ or $D_N$ are of distance $\frac{1}{4}$ of each other, which is strictly less than $\frac{1}{3}$.

Hence, if the algorithm obtains at most $\Lambda(n,\eps)$ many samples from the unknown distribution $D$, it cannot distinguish, with probability at least $\frac{2}{3}$, whether the samples are coming from $D_Y$ or $D_N$.
\end{proof}

For the proof of Lemma~\ref{cor:lbnom}, we will need the following simple claim.

\begin{cl}\label{lem:summin}
Let $\sigma: [n] \rightarrow [n]$ be a permutation and let $a_1, a_2, \ldots, a_n$ and $b_1, b_2, \ldots, b_n$ be two sets of $n$ positive real numbers. If $a_1\geq a_2\geq \dots \geq a_n$ and $b_1\geq b_2 \geq \dots \geq b_n$, then the sum $\sum\limits_{i \in [n]} \size{a_i - b_{\sigma(i)}}$ is minimized when $\sigma$ is the identity permutation.
\end{cl}


\begin{proof} First observe that if $a, b, c, d$ are four real numbers with $a\geq b$ and $c\geq d$, then the following holds:
\begin{equation}\label{eq:trivial}
    \size{a-c} + \size{b-d} \leq \size{a-d} + \size{b-c}.
\end{equation} 
The above can be proved by checking all possible orderings of the numbers $a,b,c,d$.

Once we have the above observation, we can now proceed to prove the claim. Let us consider the set of permutations that minimize $\sum\limits_{i \in [n]} \size{a_i - b_{\sigma(i)}}$. Let $\sigma$ be one such minimizing permutation that also minimizes the size for the following set $S$: 
$$S= \{(i,j) : i<j \ \mbox{and} \ \sigma(i)> \sigma(j)\}$$
Let $i$ be an index such that $\sigma(i)<\sigma(i+1)$ (such an index $i$ exists unless $\sigma$ is the identity permutation). Let $\sigma'$ be the permutation obtained from $\sigma$ by swapping $\sigma(i)$ and $\sigma(i+1)$. Then the sum $\sum\limits_{i \in [n]} \size{a_i - b_{\sigma'(i)}}$ does not increase from $\sum\limits_{i \in [n]} \size{a_i - b_{\sigma(i)}}$, because of Equation~\ref{eq:trivial}. However, the size of the set $S$ with respect to the permutation $\sigma'$ strictly decreases, and we have a contradiction.
\end{proof}

Now we present the proof of Lemma~\ref{cor:lbnom}.

\begin{proof}[Proof of Lemma~\ref{cor:lbnom}] We consider the two cases separately.

\textbf{(1) } If $\overline{D}$ is $\beta$-close to $\cP$, there exists a distribution $D_1$ in $\cP$ such that $\sum\limits_{x} \size{\overline{D}(x)- D_1(x)} \leq \beta$.
Since $\cP$ is label-invariant, any permutation of $D_1$ is also in $\cP$. Without loss of generality, let us assume that the domain $\Omega$ is a subset of $\{1, \dots, n\}$.

By Claim~\ref{lem:summin}, the permutation $\sigma$ that minimizes 
$\sum\limits_{x} \size{\overline{D}(x)- D_1(\sigma(x))} \leq \beta$ is the one that orders the $i$-th largest element of $D_1$ with the $i$-th largest element of $\overline{D}$, that is, if $x$ is the element with the $i$-th largest probability mass in $D_1$, then $\sigma(x)$ has the $i$-th largest probability mass in $\overline{D}$. 
Consider the distribution $D_1^{\sigma}$ that is defined by $D_1^{\sigma}(x) = D_1(\sigma(x))$.
Clearly, $H$ contains the largest $q^2$ elements of $D_1^{\sigma}$, and hence also $\high_{1/q^2}(D_1^{\sigma}) \subseteq H$.

As $\sum\limits_{x \in \Omega} \size{D_1^{\sigma}(x)-\overline{D}(x)} \leq \beta$, $\sum\limits_{x \in H} \size{\overline{D}(x)-\widetilde{D}(x)} \leq \alpha$ and $|\overline{D}(\Omega\setminus H)-\widetilde{D}(\Omega\setminus H)| \leq \gamma$, by the triangle inequality, we obtain
\begin{eqnarray*}
 &\sum\limits_{x \in H} \size{D_1^{\sigma}(x)-\widetilde{D}(x)} + & \size{D_1^{\sigma}(\Omega \setminus H) - \widetilde{D}(\Omega \setminus H)} \\
 \leq & \sum\limits_{x \in H} |D_1^{\sigma}(x)-\overline{D}(x)|
     + & \sum\limits_{x \in H} |\overline{D}(x)-\widetilde{D}(x)|\\
    &  &+|D_1^{\sigma}(\Omega\setminus H)-\overline{D}(\Omega\setminus H)|
       +|\overline{D}(\Omega\setminus H)-\widetilde{D}(\Omega\setminus H)|\\
 \leq & \sum\limits_{x \in H} |D_1^{\sigma}(x)-\overline{D}(x)|
       + & \sum_{x \in H} |\overline{D}(x)-\widetilde{D}(x)|\\
    &   & + \sum\limits_{x\in \Omega\setminus H} |D_1^{\sigma}(x)-\overline{D}(x)|
       +|\overline{D}(\Omega\setminus H)-\widetilde{D}(\Omega\setminus H)|\\
 = &  \sum\limits_{x\in \Omega}  |D_1^{\sigma}(x)-\overline{D}(x)|
       + & \sum\limits_{x \in H} |\overline{D}(x)-\widetilde{D}(x)|
       +|\overline{D}(\Omega\setminus H)-\widetilde{D}(\Omega\setminus H)|\\
 \leq & \alpha + \beta + \gamma &
\end{eqnarray*}

\vspace{5 pt}

\textbf{(2) } We will prove this case by contradiction. 
Let $D_1 \in \cP$ be a distribution such that  $\high_{1/q^2}(D_1) \subseteq H $ and
$\sum\limits_{x \in H}\left|D_1(x) - \widetilde{D}(x) \right| +  \left|D_1(\Omega \setminus H) - \widetilde{D}(\Omega \setminus H) \right|\leq \alpha + \beta +\gamma $. Then, as $\sum\limits_{x \in H} \size{\overline{D}(x)-\widetilde{D}(x)} \leq \alpha$, by the triangle inequality, we have \begin{equation}\label{eqn:inter}
\sum\limits_{x \in H}|D_1(x)- \overline{D}(x)| + \left|D_1(\Omega \setminus H) - \widetilde{D}(\Omega \setminus H) \right|\leq 2\alpha+ \beta + \gamma.
\end{equation}

\color{black}

Consider the distribution $\widehat{D}$ defined as follows: 
\begin{itemize}
    \item For all $x \in H$, $\widehat{D}(x) = D_1(x)$.
    \item If $D_1(H)  \geq  \overline{D}(H)$,  then for all $x \in \Omega \setminus H$, 
$$\widehat{D}(x)=\overline{D}(x) \cdot \phi, $$
where $\phi = \frac{1-D_1(H)}{1- \overline{D}(H)}.$ Notice that in this case $\phi \leq 1$.
    \item If $D_1(H)  \leq  \overline{D}(H)$, then pick the set $T\subset \Omega\setminus H$ 
    with  $|T| = 2q^2$ that minimizes $\overline{D}(T)$. Then for all $x \in T$, 
    $$\widehat{D}(x) = \overline{D}(x)+ \frac{\overline{D}(H) - D_1(H)}{2q^2}$$
    and for all $x\in \Omega\setminus (T\cup H)$, $\widehat{D}(x) = \overline{D}(x)$
\end{itemize}

Let us first prove that $\high_{1/q^2}(\widehat{D}) \subseteq H$.
In the case where $D_1(H)  \geq  \overline{D}(H)$, for all $x\in \Omega\setminus H$, $\widehat{D}(x) \leq \overline{D}(x)$. Since $\high_{1/q^2}(\overline{D})\subseteq H$, $\high_{1/q^2}(\widehat{D}) \subseteq H$.
Now, in the case where $D_1(H)  \leq  \overline{D}(H)$,
the only $x \in \Omega\setminus H$ for which $\widehat{D}(x) > \overline{D}(x)$ are those in $T$. 
Since $|\Omega| > 4q^2$, the lowest $2q^2$ elements on $\overline{D}$ must each have mass less than $\frac{1}{2q^2}$. So even if we add  $\frac{1}{2q^2}$ for any element  $x\in T$,  $\widehat{D}(x) < 1/q^2$. Hence in this case also $\high_{1/q^2}(\widehat{D}) \subseteq H$ since $\high_{1/q^2}(\overline{D})\subseteq H$ and $\high_{1/q^2}(D_1)\subseteq H$.

Now let us bound the $\ell_1$ distance between $\widehat{D}$ and $\overline{D}$. Observe that 
$$\sum\limits_{x \in \Omega \setminus H}\size{\widehat{D}(x)- \overline{D}(x)} = \left|\widehat{D}(\Omega \setminus H) - \overline{D}(\Omega \setminus H) \right|.$$  This is because, in the case where $\widehat{D}(H)  \geq  \overline{D}(H)$, we have $\widehat{D}(x) =\phi \cdot \overline{D}(x) \leq \overline{D}(x)$ for all $x \in \Omega \setminus H$.  
On the other hand, in the case where $\widehat{D}(H)  \leq  \overline{D}(H)$ then for all $x \in \Omega \setminus H$, $\widehat{D}(x) \geq \overline{D}(x)$. Thus,

\begin{eqnarray*}
\sum\limits_{x \in \Omega \setminus H}\size{\widehat{D}(x)- \overline{D}(x)} &=& \left|\widehat{D}(\Omega \setminus H) - \overline{D}(\Omega \setminus H) \right| \\
&\leq& \left|\widehat{D}(\Omega \setminus H) - \widetilde{D}(\Omega \setminus H) \right| + \left|\overline{D}(\Omega \setminus H) - \widetilde{D}(\Omega \setminus H) \right| \\
&\leq& \left|\widehat{D}(\Omega \setminus H) - \widetilde{D}(\Omega \setminus H) \right| + \gamma
\end{eqnarray*}

Also note that, from the construction of $\widehat{D}$, we have for all $x\in H$, $\widehat{D}(x) = D_1(x)$ and thus $\widehat{D}(\Omega\setminus H) = D_1(\Omega\setminus H)$. Thus,
\begin{eqnarray*}
||\widehat{D}- \overline{D}||_1 &=& \sum_{x \in H}|\widehat{D}(x) - \overline{D}(x)|  + \sum_{x \in \Omega \setminus H}|\widehat{D}(x) - \overline{D}(x)| \\
&\leq& \sum_{x \in H}|\widehat{D}(x) - \overline{D}(x)|  + \left|\widehat{D}(\Omega \setminus H) - \widetilde{D}(\Omega \setminus H) \right| +\gamma\\
&=& \left(\sum_{x \in H}|{D_1}(x) - \overline{D}(x)|  + \left|{D_1}(\Omega \setminus H) - \widetilde{D}(\Omega \setminus H) \right|\right) +\gamma \\
&&~~~~~~~~~~~~~~~~(\mbox{From the construction of $\widehat{D}$}) \\
&\leq& 2 \alpha+ \beta + 2 \gamma ~~~(\mbox{By Equation~\eqref{eqn:inter}})
\end{eqnarray*}

Moreover, as $\high_{1/q^2}(D_1) \subseteq H$ and by the construction of $\widehat{D}$, we have $\high_{1/q^2}(D_1) = \high_{1/q^2}(\widehat{D})$ and for all $x \in \high_{1/q^2}(D_1)$, $D_1(x) = \widehat{D}(x)$. 
Since we assumed that $D_1$ is in $\cP$, using Lemma~\ref{lem:lbnon}, $\widehat{D}$ is $\eps$-close to $\cP$. 
And since  $||\widehat{D}- \overline{D}||_1 \leq 2 \alpha+ \beta + 2\gamma$, we conclude that $\overline{D}$ is 
$(\eps + 2\alpha +\beta + 2\gamma)$-close to $\cP$, which is a contradiction.
\end{proof}

\color{black}

\section{Computationally efficient tolerant testers}\label{sec:consttoltest}

In this section we present a constructive variant of the tolerant tester studied in Section~\ref{sec:thm-3.1}. Here, for any two vectors $a, b \in \R^N$, we say that $a \leq b$ if $a_i \leq b_i$ holds for every $i \in [N]$. Now let us recall the definitions of {\em polyhedron} and {\em projection map}.

\begin{defi}[Polyhedron]
Let $A$ be a $M \times N$ real matrix, $b \in \R^M$ be a real vector and $Ax \leq b$ be a system of linear inequalities. The solution set $\{x \in \R^N \mid Ax \leq b \}$ of the system of inequalities is called a polyhedron. The \emph{complexity} of a polyhedron is defined as $MN$.
\end{defi}

\begin{defi}[Projection map]
Let $n$ be an integer. For all integers $N \geq n$, a \emph{projection map} is denoted as $\pi_{n} : \R^{N} \rightarrow \R^{n}$ and is defined as the projection of the points in $\R^{N}$ on the first $n$ coordinates.
\end{defi}

Before directly proceeding to our results, we first define two variants of distribution properties.

\begin{defi}[Linear property]\label{defi:linearp}
Without loss of generality, let us assume that $\Omega = [n]$.
A distribution property $\cP$ is said to be a \emph{linear property} if 
there exists a polyhedron 
$\mathcal{LP} = \left\{x \in \R^{N} \;\mid\; Ax \leq b \right\}$, where 
$A$ is a $M \times N$ real matrix and $b \in \R^{M}$ be a real vector, and $\pi_{n}\left(\mathcal{LP}\right)$ 
\footnote{Note that $\pi_{n}\left(\mathcal{LP}\right)$ will also be a polyhedron in $\R^{n}$, see, e.g., Corollary~2.5 in Chapter~2 from the book by Bertsimas and Tsitsiklis~\cite{bertsimas1997introduction}. However, the number of linear inequalities defining the property, which affects the running time of the tester, can sometimes be greatly reduced by using a projection.}
is
the set of distributions satisfying the property $\cP$, that is, for every $z := (z_{1}, \dots, z_{n}, \dots, z_{N}) \in \mathcal{LP}$, the distribution $D_{z}$, defined as
$$
    D_{z}(i) = z_{i}, \quad \forall i \in [n]
$$
satisfies the property $\cP$, and conversely, for every distribution $D$ that
satisfies $\cP$, there exists some $z \in \mathcal{LP}$ such that $D=D_z$ as defined above.
The {\em complexity} of $\cP$ is defined as $M\times N$.
\end{defi}

Similar to linear properties, we can also define properties that are feasible solutions to a system of convex constraints.

\begin{defi}[Convex property]
A distribution property $\cP$ is said to be a \emph{convex property} if $\cP$ is the set of all feasible solutions to a system of convex constraints over $D(i)$ for $i \in \Omega$, where $\Omega$ is the sample space of $D$. In other words, the set $\cP$ forms a convex set. 
\end{defi}

Now we show that some well studied label-invariant distribution properties can be represented as linear or convex properties.

\begin{rem}[An example of a linear property: Approximate uniformity property] 
A distribution $D$ over $[n]$ is said to be uniform if $D(i) = \frac{1}{n}$ for all $i \in [n]$. Let the property $\cP_{u, \eps}$ denote the set of all distributions that are $\eps$-close to the uniform distribution, where $\eps \in (0,1)$ is a parameter. 
Consider the following polyhedron $\mathcal{LP}_{u,\eps}$ in $\R^{2n}$:
\begin{align*}
    &\sum\limits_{ i \in [n]} z_{n+i} \leq \eps & \\
    & z_{i} \geq 0 &\forall i \in [2n] \\
    & - z_{n+i} \leq z_{i} - \frac{1}{n} \leq z_{n+i} &\forall i \in [n] 
\end{align*}
Now, observe that $\pi_{n}\left(\mathcal{LP}_{u,\eps}\right)$ will give us the set of distributions that are $\eps$-close to uniform, i.e., the set $\cP_{u, \eps}$ (this would serve as the linear transformation mentioned in Definition~\ref{defi:linearp}).
Also, note that approximate uniformity
property has complexity $\Oh\left(n\right)$.
\end{rem}

Now we present a property that can be expressed as a feasible solution to a system of convex inequalities, but that cannot be expressed as feasible solution to a system of linear inequalities.

\begin{rem}[An example of a convex property: Entropy property]
Let $D$ be a distribution supported on $[n]$. Given a parameter $k \in \R$, let $\cP_{E,k}$ denote the set of all distributions with entropy at least $k$. $\cP_{E,k}$ can be expressed as a convex inequality as follows:
\begin{eqnarray*}
     &&\sum\limits_{i \in [n]} D(i) \log\frac{1}{D(i)} \geq k.
\end{eqnarray*}
\end{rem}

For a distribution property $\cP$, let 
$\mathcal{CP} \subset \R^{n}$ denote the {\em geometric representation} of the set of probability distributions over the set $[n]$ that satisfy $\cP$ by considering each distribution over $[n]$ as a point in $\R^{n}$.

For all $\beta \in [0,1]$, $k \leq n$ and $a \in \R^{n}$, we define the following convex set:
$$
    \Delta \left(k,q,a, \beta\right) := \left\{ x\in \R^{d} \; : \;  \sum\limits_{i=1}^{k} \size{x_{i} - a_{i}} + \left|\sum_{j>k} x_{j} - \sum_{j > k} a_{j}\right| \leq \beta, \;\mbox{and}\; \forall i > k\; \mbox{we have}\; x_i < \frac{1}{q^2} \right\}
$$


Before proceeding to prove Theorem~\ref{thm:inf-conv}, we will first prove a more general result. We will show that if we have access to an emptiness oracle that takes a specific convex set (defined in the following statement) as input, and decides whether the convex set is empty or not, then we can design a tolerant tester for any label-invariant distribution property that takes $\widetilde{O}(\Lambda^2)$ samples and performs a single emptiness query to the oracle. The result is formally stated and proved below.

\color{black}

\begin{theo}\label{theo:tol-conv}
Let $\cP$ be a label-invariant distribution property. If there is a $(0,\eps)$-tester (non-tolerant tester) with sample complexity $\Lambda(n, \eps)$, then for any $\gamma_1$, $\gamma_2$ with $\gamma_1 < \gamma_2$ and $0 < \gamma_1 < \gamma_2 + \eps < 2$, there exists a $(\gamma_1, \gamma_2 + \eps)$-tester (tolerant tester) that takes $s = \widetilde{\Oh}(\Lambda^2)$ many samples and makes a single emptiness query to the set $\mathcal{CP} \cap \Delta(\tOh(s),\Lambda,\widetilde{D}, \beta)$, where $\widetilde{D}$ is a known probability distribution and $\beta = \gamma_1 + \frac{\gamma_2 - \gamma_1}{3}$. 
\end{theo}

\color{black}

\begin{proof}

Recall that in Step $5$ of the tolerant tester presented in Section~\ref{sec:thm-3.1}, the tester checks whether there is any distribution $D_1 \in \cP$ that satisfies the following two conditions:
$$
    \sum\limits_{x\in H}\left|D_{1}(x) - \widetilde{D}(x) \right|  + 
    \left| D_{1}(\Omega\setminus H) - \widetilde{D}(\Omega\setminus H)\right|
    \leq  26\eta' + \zeta
$$
and 
$$
    \high_{1/q^2}(D_1) \subseteq H
$$
where $\zeta = \gamma_1$, $\eta = {\gamma_2 - \gamma_1}$, $\eta = {\gamma_2 - \gamma_1}$ and $\eta'=\frac{\eta}{64}$.
The set $H$ and the distribution $\widetilde{D}$ are defined in the tolerant tester presented in Section~\ref{sec:thm-3.1}.

Without loss of generality, we can assume that $H = \left\{1, \ldots, \size{H}\right\}$. 
Therefore, in order to perform Step $5$ of the tolerant tester, the following equations are needed to be satisfied:
\begin{eqnarray}
    && D_1 \in \mathcal{CP}\\
    && D_1 \in \Delta \left(\size{H},q,\widetilde{D}, 26 \eta' + \zeta\right) 
\end{eqnarray}

\color{black}

We now present the tolerant 
$(\gamma_1, \gamma_2 + \eps)$-tester in its entirety, that is,  a $(\zeta, \zeta + \eps + \eta)$-tester for the property $\mathcal{P}$, where $\zeta = \gamma_1$, $\eta = {\gamma_2 - \gamma_1}$, and $\eta'=\frac{\eta}{64}$.

\begin{enumerate}

    \item Draw $W = \Oh\left(\frac{q^2}{\eta'}\log q \right)$ many samples from the distribution $D$. Let $S\subseteq \Omega$ be the set of (distinct) samples obtained.
    
    \item  Draw additional $\Oh \left(\frac{W}{\eta'^2}\log W \right)$ many samples $Z$ to estimate the value of $D(x)$ for all $x\in S$.
    
    \item Construct a set $H$ as the union of $S$ and $q^2$ many arbitrary elements from $\Omega \setminus (S \cup Z)$. 
    
    \item Define a distribution $\widetilde{D}$ such that, for $x \in H$, 
    $$\widetilde{D}(x)=\frac{\#~ x~\mbox{in the multi-set}~ Z}{|Z|}.$$
    
   And for each $x\in \Omega\setminus H$, 
    $$\widetilde{D}(x) =\frac{1-\sum\limits_{x \in H}\widetilde{D}(x)}{\size{\Omega}-\size{H}}.$$

    \item
    
    If there exists a distribution $D_1 \in \mathcal{CP} \cap \Delta \left(\size{H},q,\widetilde{D}, 26\eta' + \zeta\right)$, 
     then ACCEPT $D$.

    \color{black}
    
    \item If there does not exist any distribution $D_1$ that passes Step $5$, then REJECT $D$. 
    
    \color{black}
\end{enumerate}

Observe that the sample complexity of the tester is  $\Oh\left(\frac{q^2}{\eta^2}\log^2 q\right) = \widetilde{\Oh}(\Lambda^2)$ in addition to a single emptiness query to the set $\mathcal{P} \in \mathcal{CP} \cap \Delta \left(\size{H},q,\widetilde{D}, 26 \eta' + \zeta\right)$ in Step $5$. The correctness proof of the above tester follows from the correctness argument presented in Section~\ref{sec:thm-3.1}.
\end{proof}

\subsection{Emptiness checking when \texorpdfstring{$\mathcal{P}$}~~is a linear property: Proof of Theorem~\ref{thm:inf-conv}}

Now we proceed to analyze the time complexity of the $(\gamma_1, \gamma_2 + \eps)$-tester described in Theorem~\ref{theo:tol-conv} when $\mathcal{P}$ is also a linear property. Recall 
that as $\cP$ is a linear property, 
there exists a polyhedron 
$\mathcal{LP} = \left\{x \in \R^{N} \;\mid\; Ax \leq b \right\}$, where 
$A$ is a $M \times N$ real matrix and $b \in \R^{M}$ be a real vector, and $\pi_{n}\left(\mathcal{LP}\right)$ is
the set of distributions satisfying the property $\cP$.
(See Definition~\ref{defi:linearp})

Now in Observation~\ref{obs:obs1}, we show that checking emptiness of $\pi_n(\mathcal{LP}) \cap \Delta\left(\size{H},q, \widetilde{D}, 26 \eta' + \zeta\right)$ is equivalent to testing the feasibility of a family of inequalities.

\begin{obs}\label{obs:obs1}
Without loss of generality, assume that $H=\{1, \ldots, |H|\}$ and $\Omega = \{1, \ldots, n\}$. 
Checking emptiness of $\pi_n(\mathcal{LP}) \cap \Delta(\size{H}, q,\widetilde{D}, 26\eta' + \zeta)$ is equivalent to testing the feasibility of the following set of inequalities:
\begin{align}
     & Az \leq b & \\
    &\sum\limits_{i=1}^{|H|}\left|z_i - \widetilde{D}(i) \right|  + 
    \left|\sum_{i= |H|+1}^{n} z_{i} - \sum_{i= |H| + 1}^{n} \widetilde{D}(i)\right| \leq  26\eta' + \zeta & \label{eqn:testhighlinear} \\
    & z_i < \frac{1}{q^2} &\forall i \in [n] \setminus \{1, \ldots, \size{H}\}
\end{align}
\end{obs}

Note that the inequality in Equation~\eqref{eqn:testhighlinear} can be expressed as the following set of linear inequalities using slack variables $z_{N+i}$ for all $i \in \left[\size{H}+1\right]$:
\begin{align*}
    &\sum\limits_{i=1}^{|H|} z_{N+i} \; + \; z_{N+\size{H}+1}\leq 26 \eta' + \zeta & \\
    & z_{N+i} \geq 0 &\forall i \in \left[\size{H}+1\right] \\
    & - z_{N+i} \leq z_i - \widetilde{D}(i) \leq z_{N+i} &\forall i \in \left[\size{H}\right]\\
    & - z_{N+\size{H}+1} \leq \sum\limits_{i = \size{H}+1}^{n} z_i - \sum\limits_{i= \size{H}+1}^{n} \widetilde{D}(i) \leq z_{N+\size{H}+1}
\end{align*}

Therefore checking the emptiness of $\pi_n(\mathcal{LP}) \cap \Delta\left(\size{H},q,\widetilde{D}, 26 \eta' + \zeta\right)$ is equivalent to checking the feasibility of the following set of linear inequalities:
\begin{align*}
     & Az \leq b \\
     &\sum\limits_{i=1}^{|H|} z_{N+i} + z_{N+\size{H}+1}\leq 26 \eta' + \zeta \\
    & z_{N+i} \geq 0 &\forall i \in [\size{H}+1] \\
    & - z_{N+i} \leq z_i - \widetilde{D}(i) \leq z_{N+i}  &~\forall i \in [\size{H}] \\
    & - z_{N+\size{H}+1} \leq \sum\limits_{i=\size{H}+1}^{n} z_i - \sum\limits_{i =\size{H}+1}^{n}\widetilde{D}(i) \leq z_{N+\size{H}+1}\\
    & z_i < \frac{1}{q^2} &\forall i \in [n] \setminus \{1, \ldots, \size{H}\}
\end{align*}
The feasibility of the above set of linear inequalities can be solved in a polynomial time in the complexity of the polyhedron, that is, in a polynomial time in $N$ and $M$, using the Ellipsoid Method, where recall that $A$ is a $M\times N$ real matrix (see, e.g., \cite{bertsimas1997introduction, matousek2007understanding}). Thus, we have an efficient $(\gamma_1, \gamma_2 + \eps)$-tester for $\cP$, that runs in time polynomial in the complexity of the label-invariant linear property
$\cP$. This concludes the proof of Theorem~\ref{thm:inf-conv}.

\color{black}

\section{Sample complexity of testing non-concentrated label-invariant properties}\label{sec:nonconclb}

In this section we first prove a lower bound of $\Omega(\sqrt{n})$ on the sample complexity of non-tolerant testing of any non-concentrated label-invariant property. Then we proceed to prove a tolerant lower bound of $\Omega(n^{1-o(1)})$ samples for such properties in Section~\ref{sec:sym-tol}.

\subsection{Non-tolerant lower bound (Proof of Theorem~\ref{thm:nonconclb} for label-invariant properties)}\label{sec:sym-ntol}

Here we first prove a lower bound result analogous to Theorem~\ref{thm:nonconclb}, where the properties are non-concentrated and  label-invariant. In Section~\ref{sec:thm-noncon}, we discuss why the proof of Theorem~\ref{theo:sym} does not directly work for Theorem~\ref{thm:nonconclb}, and then prove  Theorem~\ref{thm:nonconclb} using a different argument. 

\begin{theo}[Analogous result of Theorem~\ref{thm:nonconclb} for non-concentrated label-invariant properties]\label{theo:sym}
Let $\mathcal{P}$ be any $(\alpha, \beta)$-non-concentrated label-invariant distribution property, where $0< \alpha \leq \beta < \frac{1}{2}$. For $\eps$ with $0< \eps < \alpha$,  any $(0,\eps)$-tester for property $\cP$ requires $\Omega(\sqrt{n})$ many samples, where $n$ is the size of the support of the distribution.
\end{theo}

\begin{proof}
     Let us first consider a distribution $D_{yes}$ that satisfies the property. Since $\mathcal{P}$ is an $(\alpha, \beta)$-non-concentrated property, by Definition~\ref{defi:non-conc-prop}, $D_{yes}$ is an $(\alpha, \beta)$-non-concentrated distribution. From $D_{yes}$, we generate a distribution $D_{no}$ such that the support of $D_{no}$ is a subset of that of $D_{yes}$, and $D_{no}$ is $\eps$-far from $\cP$. Hence, if we apply a random permutation over the elements of $\Omega$, we show that $D_{yes}$ and $D_{no}$ are indistinguishable, unless we query for  $\Omega(\sqrt{n})$ many samples. Below we formally prove this idea.
    
We will partition the domain $\Omega$ into two parts, depending on the probability mass of $D_{yes}$ on the elements of $\Omega$. 
Given the distribution $D_{yes}$, let us first order the elements of $\Omega$ according to their probability masses. In this ordering, let $L$ be the smallest $2\beta n$ elements of $\Omega$. We denote $\Omega \setminus L$ by $H$. Before proceeding further, note that the following observation gives an upper bound on the probabilities of the elements in $L$.

\begin{obs}\label{obs:LH}
For all $x \in L$, $D_{yes}(x) \leq \frac{1-2\alpha}{1-2\beta}\frac{1}{n}$. 
\end{obs}

\begin{proof}[Proof of Observation~\ref{obs:LH}]
By contradiction, assume that there exists $x \in L$ such that  $D_{yes}(x)>\frac{1-2\alpha}{1-2\beta}\frac{1}{n}$. This implies, for every $y \in H$, that $D_{yes}(y) > \frac{1-2\alpha}{1-2\beta}\frac{1}{n}$. So, $$ 1= \sum _{x\in \Omega} D_{yes}(x)=\sum_{x \in L} D_{yes}(x) +\sum_{y \in H}D_{yes}(y)> D_{yes}(L)+\size{H}\frac{1-2\alpha}{1-2\beta}\frac{1}{n}.$$
As $\size{L}=2\beta n$ and $D_{yes}$ is an $(\alpha,\beta)$-non-concentrated distribution, $D_{yes}(L)\geq 2\alpha$. Also, $\size{H}=(1-2\beta)n$. Plugging these into the above inequality, we get a contradiction.
\end{proof}

Note that Observation~\ref{obs:LH} implies that if $S$ is a multi-set of  $o\left(\sqrt{\frac{1-2\beta}{1-2\alpha}n}\right)$ samples from $D_{yes}$, then with probability $1-o(1)$, no element from $L$ appears in $S$ more than once. 
Now using the distribution $D_{yes}$ and the set $L$, let us define a distribution $D_{no}$ such that $D_{no}$ is $\eps$-far from $\mathcal{P}$. Note that $D_{no}$ is a distribution that comes from a distribution over a set of distributions, all of which are not $(\alpha, \beta)$-non-concentrated. The distribution $D_{no}$ is generated using the following random process:

\begin{itemize}
    \item  We partition $L$ randomly into two equal sets of size $\beta n$. Let the sets be $\{x_1, \dots, x_{\beta n}\}$ and $\{y_1, \dots, y_{\beta n}\}$. 
    We pair the elements of $L$ randomly into $\beta n$ pairs. Let $(x_1,y_1),\ldots,(x_{\beta n},y_{\beta n})$ be a random pairing of the elements in $L$, which is represented as $P_L$, that is, $P_L = \{(x_1,y_1),\ldots,(x_{\beta n},y_{\beta n})\}$. 
   
    \item The probability mass of $D_{no}$ at $z$  is defined as follows:
    \begin{itemize}
        \item If $z \not\in L$, then $D_{no}(z)=D_{yes}(z)$.
        
        \item For every pair $(x_i, y_i) \in P_L$, $D_{no}(x_i)=D_{yes}(x_i)+D_{yes}(y_i)$, and $D_{no}(y_i)=0$.
    \end{itemize}
\end{itemize}

We start by observing that the distribution $D_{no}$ constructed above is supported on a set of at most $(1-\beta)n$ elements. So, any distribution $D_{no}$ constructed using the above procedure is $\eps$-far from satisfying the property $\mathcal{P}$ for any $\eps < \alpha$.

We will now prove that $D_{yes}$ and $D_{no}$ both have similar distributions over the sequences of samples. More formally, we will prove that any algorithm that takes $o(\sqrt{n})$ many samples, cannot distinguish between $D_{yes}$ from $D_{no}$ with probability at least $\frac{2}{3}$.

Since any $D_{no}$ produced using the above procedure has exactly the same probability mass on the elements in $H$ as $D_{yes}$, any tester that distinguishes between $D_{yes}$ and $D_{no}$ must rely on samples obtained from $L$. Recall that the algorithm is given a uniformly random permutation of the distribution. Since $\supp(D_{no})  \subset \supp(D_{yes})$ (particularly, $\supp(D_{no}) \cap L \subset \supp(D_{yes}) \cap L$), it is not possible to distinguish between $D_{yes}$ and $D_{no}$, unless an element of $L$ appears at least twice. Otherwise, as in the proof of Lemma~\ref{lem:lbnon}, the elements drawn from $L$ are distributed identically to a uniformly random non-repeating sequence. But
observe that $D_{yes}(i)=\Oh(\frac{1}{n}) $ and $D_{no}(i)=\Oh(\frac{1}{n})$ when $i$ is in $L$. Thus any sequence of $o(\sqrt{n})$ samples will provide only a distance of $o(1)$ between the two distributions, completing the proof.
\end{proof}


\subsection{Tolerant lower bound (Proof of Theorem~\ref{thm:tol-sym-non})}\label{sec:sym-tol}
\begin{theo}[Theorem~\ref{thm:tol-sym-non} formalized]\label{theo:tol-sym-lb}
Let $\mathcal{P}$ be any  $(\alpha, \beta)$-non-concentrated label-invariant distribution property, where $0< \alpha \leq \beta < \frac{1}{2}$. For any constant $\eps_1$ and $\eps_2$ with $0< \eps_1 < \eps_2 < \alpha $, any $(\eps_1,\eps_2)$-tester for $\cP$ requires $\Omega({n}^{1-o(1)})$ samples, where $n$ is the size of the support of the distribution.
\end{theo}

To prove the above theorem, we recall some notions and a theorem from Valiant's paper on a lower bound for the sample complexity of tolerant testing of symmetric properties ~\cite{valiant2011testing}. These definitions refer to invariants of distributions, which are essentially a generalization of properties.


\begin{defi}\label{defi:weakcont}
Let $\Pi:\cD_n \rightarrow \R$ denote a real-valued function over the set $\cD_n$ of all distributions over $[n]$. 
\begin{enumerate}
    \item $\Pi$ is said to be \emph{label-invariant} if for any $D \in \cD_n$ the following holds: $\Pi(D)=\Pi(D_\sigma)$ for any permutation $\sigma:[n]\rightarrow [n]$. 
    
    \item  For any $\gamma ,\delta$ with $\gamma \geq 0$ and $\delta \in [0,2]$, $\Pi$ is said to be \emph{$(\gamma,\delta)$-weakly-continuous} if for all distributions $p^+,p^-$ satisfying $||p^+ - p^-||_1 \leq \delta$, we have $|\Pi(p^+) - \Pi(p^-)| \leq \gamma$.
\end{enumerate}
\end{defi}

For a property $\cP$ of distributions, we define
$\Pi_{\cP}:\cD_n \rightarrow [0,2]$ with respect to property $\cP$ as follows:
$$\mbox{For}~D \in \cD_n, \Pi_{\cP}(D):=\mbox{the distance of $D$ from $\cP$.
} $$

From the triangle inequality property of $\ell_1$ distances, $\Pi_{\cP}$ (which refers to the distance function from the property $\cP$) is $(\gamma,\gamma)$-weakly continuous, for any $\gamma \in [0,2]$.

\begin{theo}[Low Frequency Blindness~\cite{valiant2011testing}]\label{theo:lowfreq}
Consider a function $\Pi:\cD_n \rightarrow \cR$ that is label-invariant and $(\gamma, \delta)$-weakly-continuous, where $\gamma \geq 0$ and $\delta\in [0,2]$. Let there exist two distributions $p^+$ and $p^-$ in $\cD_n$ with  $n$ being the size of their supports, such that $\Pi(p^+)> b$, $\Pi(p^-)< a$, and they are identical  for any index occurring  with  probability  at  least $\frac{1}{n}$ in either distribution, where $a,b \in \R$. Then any tester that has sample access to an unknown distribution $D$ and distinguishes between $\Pi(D) > b - \gamma$ and $\Pi(D) < a + \gamma$, requires $\Omega(n^{1-o_{\delta}(1)})$ many samples from $D$~\footnote{$o_\delta(\cdot)$ suppresses a term in $\delta$.}.
\end{theo}

\color{black}

Note that in Theorem~\ref{theo:lowfreq}, we have assumed that $p^+$ and $p^-$ are identical for any index that has probability mass at least $\frac{1}{n}$. We can actually replace this condition to $\Oh(\frac{1}{n})$ by adding $\Oh(n)$ many ``dummy elements'' to the support of $p^+$ and $p^-$. Now we are ready to prove Theorem~\ref{theo:tol-sym-lb}.

\begin{proof}[Proof of Theorem~\ref{theo:tol-sym-lb}]

Consider $\Pi_{\cP}$ as defined above. As $\cP$ is a label-invariant property, the function $\Pi_{\cP}$ is also label-invariant. We have already noted that $\Pi_{\cP}$ is $(\gamma,\gamma)$-weakly continuous as ``\emph{distance from a property}'' satisfies the triangle inequality, for any $\gamma \in [0,2]$. Now recall that the distributions $D_{yes}$ and $D_{no}$ considered in the proof of Theorem~\ref{theo:sym}. The probability mass of each element in the support of $D_{yes}$ and $D_{no}$ is $\Oh(\frac{1}{n})$. Note that $D_{yes}$ is in $\cP$ and $D_{no}$ is $\eps$-far from $\cP$, for any $\eps<\alpha$, and both of them have a support size of $\Theta(n)$. Here we take $\eps >\eps_2$. Now, we apply Theorem~\ref{theo:lowfreq} with $a=0$, some $b <\eps$ and $\gamma$ with $\gamma <\min\{\eps_1, \eps -\eps_2\}$. Observe that this completes the proof of Theorem~\ref{theo:tol-sym-lb}.
\end{proof}

\section{Sample complexity of non-concentrated properties (Proof of Theorem~\ref{thm:nonconclb})} \label{sec:thm-noncon}


\begin{theo}[Theorem~\ref{thm:nonconclb} formalized]\label{thm:non-sym}
Let $\mathcal{P}$ be any $(\alpha, \beta)$-non-concentrated distribution property for $0 < \alpha < \beta < \frac{1}{2}$. For any $\eps$ with $0 < \eps < \alpha$,  any $(0,\eps)$-tester for $\cP$ requires $\Omega(\sqrt{n})$ many samples, where $n$ is the size of the support of the distribution.
\end{theo}

\paragraph*{Why does the proof of Theorem~\ref{theo:sym} work only for label-invariant properties?}
Note that the proof of Theorem~\ref{theo:sym} crucially uses the fact that the property $\cP$ is label-invariant. Recall that, while constructing $D_{no}$ from $D_{yes}$,  for each $i \in [\beta n]$, moving the masses of both $x_i$ and $y_i$ in $D_{yes}$ to $x_i$ to produce $D_{no}$ is possible as the property $\cP$ is label-invariant. Due to this feature, we can apply a random permutation over $\Omega$, and still the permuted distribution will behave identically with respect to $\cP$. After applying the random permutation, the samples coming from $D_{yes}$ and $D_{no}$ are indistinguishable as long as there are no collisions among the elements in $L$, which is the case when we take $o(\sqrt{n})$ samples. However, this technique does not work when the property is not label-invariant, as the value of the distribution with respect to $\cP$ may not be invariant under the random permutation over $\Omega$. This requires a new argument; although the proof is similar in spirit to the proof of Theorem~\ref{theo:sym}, there are some crucial differences, and we present the proof next. In order to prove Theorem~\ref{thm:non-sym}, instead of moving the masses of both $x_i$ and $y_i$ in $D_{yes}$ to $x_i$ to produce $D_{no}$, we randomly move the sum to either $x_i$ or $y_i$, with probability proportional to the masses of $x_i$ and $y_i$.
    
\subsection{Proof of Theorem~\ref{thm:non-sym}}\label{sec:genlb}

The proof of Theorem~\ref{thm:non-sym} starts off identically to the proof of Theorem~\ref{theo:sym}, but there is a departure in the construction of $D_{yes}$ and $D_{no}$.


Let us first consider $D_{yes}$, $L$ and $P_L$ as discussed in the proof of Theorem~\ref{theo:sym}, only here we cannot and will not pass $D_{yes}$ through a random permutation. The difference starts from the description of the distribution $D_{no}$. In fact, $D_{no}$ will be randomly chosen according to a distribution over a set of distributions, all of which are not $(\alpha, \beta)$-non-concentrated. 
The distribution $D_{no}$ is generated using the following random process:

\begin{itemize}
    \item  We partition $L$ arbitrarily into two equal sets of size $\beta n$. Let the sets be $\{x_1, \dots, x_{\beta n}\}$ and $\{y_1, \dots, y_{\beta n}\}$. 
    We pair the elements of $L$ arbitrarily into $\beta n$ pairs. Let $(x_1,y_1),\ldots,(x_{\beta n},y_{\beta n})$ be an arbitrary pairing of the elements in $L$. Let $P_L$ be the set of pairs. So $P_L = \{(x_1,y_1),\ldots,(x_{\beta n},y_{\beta n})\}$. We refer to $x_i$ and $y_i$ as the elements corresponding to each other with respect to $P_L$, and denote $\pi(x_i)=y_i$ and $\pi(y_i)=x_i$.
    
   
    \item The probability mass of $D_{no}$ at $z$  is defined as follows:
    \begin{itemize}
        \item If $z \not\in L$, then $D_{no}(z)=D_{yes}(z)$.
        
        \item For every pair $(x_i, y_i) \in P_L$, use independent random coins and  
        \begin{itemize}
            \item With probability $\frac{D_{yes}(x_i)}{D_{yes}(x_i) + D_{yes}(y_i)}$, set $D_{no}(x_i) = D_{yes}(x_i) + D_{yes}(y_i)$ and $D_{no}(y_i)= 0$.
            
            \item With the remaining probability, that is, with probability $\frac{D_{yes}(y_i)}{D_{yes}(x_i) + D_{yes}(y_i)}$, set $D_{no}(x_i)\\ = 0$ and $D_{no}(y_i) = D_{yes}(x_i) + D_{yes}(y_i)$.
        \end{itemize}
    \end{itemize}
\end{itemize}

Observe that any $D_{no}$ constructed by the above procedure is supported over a set of at most $(1-\beta)n$ elements. So any distribution $D_{no}$ constructed using the above procedure is $\eps$-far from satisfying the property $\mathcal{P}$, for any $\eps < \alpha$. But since any $D_{no}$ produced using the above procedure has exactly the same probability mass on elements in $H$ as $D_{yes}$, any tester that distinguishes between $D_{yes}$ and $D_{no}$ must rely on samples obtained from $L$. However, we can prove that unless we receive two samples from the same pair in $L$ (which occurs with low probability), the sample sequence cannot distinguish $D_{yes}$ from $D_{no}$.

Note that there is an upper bound of $\Oh(\frac{1}{n})$ on the probability mass of any element in $L$.  In fact, for any pair $(x_i, y_i) \in P_L$, the total probability mass of the pair is at most $\Oh(\frac{1}{n})$.

\begin{obs}[Follows immediately from Observation~\ref{obs:LH}]\label{obs:LH-no}
For all pairs $(x_i, y_i)  \in P_L$, $D_{no}(x_i) + D_{no}(y_i) \leq 2\frac{1-2\alpha}{1-2\beta}\frac{1}{n}$. 
Also note that $D_{no}(x_i) + D_{no}(y_i) = D_{yes}(x_i) + D_{yes}(y_i)$ with probability $1$ over the construction of $D_{no}$.
\end{obs}

From Observation~\ref{obs:LH-no}, observe that if $S$ is a multi-set of  $o\left(\sqrt{\frac{1-2\beta}{1-2\alpha}n}\right)$ samples from $D_{yes}$, then with probability $1-o(1)$, no two elements in $S$ (identical or not) are from the same pair in $P_L$. The same holds for $D_{no}$ as well. Given that no two elements in $S$ are from the same pair in $P_L$, we will now prove that $D_{yes}$ and $D_{no}$ have the same distributions over sample sequences. This implies that, for a sequence of $o\left(\sqrt{\frac{1-2\beta}{1-2\alpha}n}\right)$ many samples, $D_{yes}$ and $D_{no}$ induce
distributions over samples sequences that have $o(1)$ variation distance
from each other.

Note that under the condition that at most one element is drawn from any pair $(x_i, y_i) \in P_L$, the probability that the sample is $x_i$ instead of $y_i$ is equal to $\frac{D_{yes}(x_i)}{D_{yes}(x_i) + D_{yes}(y_i)}$, irrespective of whether the distribution is $D_{yes}$ or $D_{no}$. So, we have the following lemma.

\begin{lem}\label{lem:imp}
Let $a_1,...,a_q$ be a sequence of $q$  elements, where no element of $L$ appears twice, additionally containing no two elements from the same pair in $P_L$ (elements of $H$ can appear freely).
Then
$$
    \Pr_{s_1, \dots, s_q \sim D_{yes}}[(s_1, \ldots, s_q)=(a_1, \ldots, a_q)] = \Pr_{s_1, \dots, s_q \sim D_{no}}[(s_1, \ldots, s_q)=(a_1, \ldots, a_q)].
$$    
\end{lem}

\begin{proof}
Let us begin by defining an event $\cE$ as follows: 
$$\cE := \mbox{no element of L appears twice, and no two elements from the same pair appear}.
$$
Observe that we will be done by proving 
\begin{equation}\label{eqn:ind}
    \Pr_{s_1, \dots, s_q \sim D_{yes}}[s_i=a_i~\mbox{for each $i \in [q]$}~|~\cE] = \Pr_{s_1, \dots, s_q \sim D_{no}}[s_i=a_i~\mbox{for each $i \in [q]$}~|~\cE].
\end{equation}

We will prove this by using induction over $q$. Let us assume that we have generated samples $s_1=a_1, \ldots, s_k=a_k$ from the unknown distribution, where $1\leq k <q$. Let $X_k = \{s_1, \ldots, s_k\} \cap L$ be the samples we have seen until now from $L$, and $X_k'=\{\pi (x) : x \in X_k\}$. By the induction hypothesis, assume that Equation~\eqref{eqn:ind} holds for each $q$ with $q \leq k$. We will show that Equation~\eqref{eqn:ind} holds for $q=k+1$.

To do so, let us now define two distributions $D_{yes}^{k+ 1}$ and $D_{no}^{k+1}$ as follows, and consider a claim (Claim~\ref{cl:yesno}) about them.

$$D_{yes}^{k+1} (x) = \Pr _{{s_1, \dots, s_q \sim D_{yes}}}\left[s_{k+1}=x \mid \cE \ \mbox{and} \ s_i=a_i \ \mbox{for} \ i \leq k\right].$$

Similarly, 
$$D_{no}^{k+1} (x) = \Pr_{{s_1, \dots, s_q \sim D_{no}}}[s_{k+1}=x \mid \cE \ \mbox{and} \ s_i=a_i \ \mbox{for} \ i \leq k].$$

\begin{cl}\label{cl:yesno}
$D_{yes}^{k+1} (x) = D_{no}^{k+1} (x)~\mbox{for every $x \in \Omega$}.$ 
\end{cl} 
\begin{proof}

We prove the claim separately when $x \in X_k \cup X_k'\subseteq L$, $x \in L \setminus( X_k \cup X_k')$, and $x \notin L$.
 \begin{description}
 \item[(i) $x \in  X_k \cup X_k'$:] $D_{yes}^{k+1}(x)=D_{no}^{k+1}(x)=0$. This follows from the condition that no element of $L$ appears twice, additionally containing no two elements of the same pair.
 \item[(ii) $x \in L \setminus(X_k \cup X_k')$:]  As $D_{yes}^{k+1}(x)=D_{no}^{k+1}(x)=0$ for every $x \in X_k \cup X_k'$, we have the followings for each $x \in L \setminus( X_k \cup X_k')$.

Assume that $x=x_i \in L \setminus (X_k \cup X_k')$ for some $i\in [\beta n]$ (using the notation defined for the partition of $L$ into pairs while we have described the random process for generating $D_{no}$). The argument for the case where $x=y_j$ for some $j \in [\beta n]$ is analogous to this.

Under $D_{yes}$, a direct calculation gives the probability for obtaining $x=x_i \in L \setminus (X_k \cup X_k')$ as the $(k+1)$-th sample $s_{k+1}$.
 \begin{eqnarray*}
    D_{yes}^{k+1}(x) &=& D_{yes}(x~|~x \notin X_k \cup X_k')\\
    &=& \frac{D_{yes}{(x)}}{1-\sum\limits_{y \in X_k \cup X_k'}D_{yes}(y)}\\
    &=&
    \frac{D_{yes}{(x)}}{1-\sum\limits_{y \in X_k }\left(D_{yes}(y)+ D_{yes}(\pi (y))\right)},
\end{eqnarray*}

Let us now consider $D_{no}$. Note that $x_i \in L \setminus (X_k \cup X_k')$, and neither $x_i$ nor $y_i$ is present in the set of first $k$ samples $\{s_1,\ldots,s_k\}$. So, the probability of getting $s_1,\ldots,s_k$ as the sequence of first $k$ samples is completely independent of how $D_{no}(x_i)$ and $D_{no}(y_i)$ are assigned while generating $D_{no}$, that is, whether we chose $D_{no}(x_i)$ to be $D_{yes}(x_i)+D_{yes}(y_i)$, or chose it to be zero (and made $D_{no}(y_i)$ equal to $D_{yes}(x_i)+D_{yes}(y_i)$ instead). That is, even when conditioned on the event that $s_1,\ldots,s_k$ is the sequence of first $k$ samples, the probability that $D_{no}(x_i)$ is $D_{yes}(x_i)+D_{yes}(y_i)$ is $\frac{D_{yes}(x_i)}{D_{yes}(x_i)+D_{yes}(y_i)}$.  Note that $D_{no}(x_i)$ is $0$ with probability $\frac{D_{yes}(y_i)}{D_{yes}(x_i)+D_{yes}(y_i)}$.

Now we can calculate the probability of obtaining $x=x_i \in L$ as the $(k+1)$-th sample $s_k$ from the corresponding conditional probabilities.
\begin{eqnarray*}
    D_{no}^{k+1}(x) &=&
    D_{no}\left(x~|~x \notin X_k \cup X_k'\right)\\
    &=& \frac{D_{yes}(x_i)+D_{yes}{(y_i)}}{1-\sum\limits_{y \in X_k \cup X_k'}D_{no}(y)} \cdot \frac{D_{yes}(x_i)}{D_{yes}(x_i)+D_{yes}{(y_i)}}\\
    &=& \frac{D_{yes}{(x)}}{1-\sum\limits_{y \in X_k }\left(D_{no}(y)+ D_{no}(\pi (y))\right)}.
\end{eqnarray*}

From the construction of $D_{yes}$ and $D_{no}$, for each $y \in L$, $D_{yes}(y)+D_{yes}(\pi (y))=D_{no}(y)+D_{no}(\pi (y))$. As $X_k \subseteq L$, 
$$ \sum\limits_{y \in X_k }\left(D_{yes}(y)+ D_{yes}(\pi (y))\right)=\sum\limits_{y \in X_k }\left(D_{no}(y)+ D_{no}(\pi (y))\right).$$ 




Hence, we have $D^{k+1}_{yes}(x)=D^{k+1}_{no}(x)$. 
\item[(iii) $x \notin L$:] Recall that for any $x \notin L$, $D_{yes}(x)=D_{no}(x)$. Proceeding in similar fashion to $D_{yes}^{k+1}(x)$ in Case $(ii)$, we conclude that $D_{yes}^{k+1}(x)=D_{no}^{k+1}$.
\end{description}
\end{proof}

\remove{To do so, we need another observation about $X_k$.

\begin{obs}\label{obs:two}
$\sum\limits_{x \in X_k}D_{yes}^k(x)+D_{yes}^k(\pi (x)) =\sum\limits_{x \in X_k}D_{no}^k(x)+D_{no}^k(\pi (x))$.
\end{obs}  
Recall that $P_L=\{(x_1,y_1),\ldots,(x_{\beta n},y_{\beta n})\}$. Note that $\pi(x_i)=y_i$ and $\pi(y_i)=x_i$. By the construction of $D_{no}$ from $D_{yes}$, $ D_{yes}^k(x_i) + D_{yes}^k(y_i)= D_{no}^k(x_i) + D_{no}^k(y_i)$ for each $i \in [\beta n]$. So, Observation~\ref{obs:two} follows.  

Now we argue that $D_{yes}^k$ and $D_{no}^k$ are identical. Note that $D_{no}^k$ is a distribution generated from $D_{yes}^k$ in a probabilistic way. For each $x_i$, 
$D_{no}^k(x_i)=\frac{D_{yes}^k(x_i)}{D_{yes}^k(x_i)+ D_{yes}^k(y_i)} \cdot (D_{yes}^k(x_i)+ D_{yes}^k(y_i))={D_{yes}^k(x_i)}$. Similarly,  $D_{no}^k(y_i)=D_{yes}^k(y_i)$. 

By Observations~\ref{obs:one} and~\ref{obs:two} along with the fact that  $D_{yes}^k$ and $D_{no}^k$ are identical, $D_{yes}^k(x)=D_{no}^k(x)$ for each $x \notin X_k \cup X_k'$. 
}

Now we are ready to prove Equation~\eqref{eqn:ind} for $q=k+1$. 
\begin{eqnarray*}
    &&\Pr_{s_1, \dots, s_{k+1} \sim D_{yes}}[s_i=a_i~\mbox{for each $i \in [k+1]$}~|~\cE]\\
    &=&\Pr_{s_1, \dots, s_{k+1} \sim D_{yes}}[s_i=a_i~\mbox{for each $i \in [k]$}~|~\cE] \cdot \Pr[s_{k+1}=a_{k+1}~|~\cE~\mbox{and}~s_i=a_i~\mbox{for each $i \in [k]$}]\\
    &=& \Pr_{s_1, \dots,s_k \sim D_{yes}}[s_i=a_i~\mbox{for each $i \in [k]$}~|~\cE]\cdot D_{yes}^{k+1}(a_{k+1})~~~~~~~~~~~~~~~~(\mbox{By the definition of $D_{yes}^{k+1}$}) \\
    &=& \Pr_{s_1, \dots, s_{k} \sim D_{no}}[s_i=a_i~\mbox{for each $i \in [k]$}~|~\cE]\cdot D_{no}^{k+1}(a_{k+1}) \\
    &&~~~~~~~~~~~~~~~~~~~~~~~~~~~~~~~~~~~~~~~~~~~~~(\mbox{By the induction hypothesis and Claim~\ref{cl:yesno}, respectively})\\
    &=&   \Pr_{s_1, \dots, s_k \sim D_{no}}[s_i=a_i~\mbox{for each $i \in [k]$}~|~\cE] \cdot \Pr[s_{k+1}=a_{k+1}~|~\cE~\mbox{and}~s_i=a_i~\mbox{for each $i \in [k]$}]\\
    &=& \Pr_{s_1, \dots, s_{k+1} \sim D_{no}}[s_i=a_i~\mbox{for each $i \in [k+1]$}~|~\cE].
\end{eqnarray*}
\end{proof}

Following the construction of $D_{yes}$ and $D_{no}$, we know that the two distributions differ only on the elements of $L$. Moreover, following Observation~\ref{obs:LH-no}, we know that if we take $o\left(\sqrt{\frac{1-2\beta}{1-2\alpha}n}\right)$ many samples, then with probability $1- o(1)$, neither any element of $L$ will appear more than once nor two elements of same pair in $P_L$ will appear. Under these two conditions, Lemma~\ref{lem:imp} states that $D_{yes}$ and $D_{no}$ will appear to be the same. Thus we can say that any ($0, \epsilon$)-tester that receives $o\left(\sqrt{\frac{1-2\beta}{1-2\alpha}n}\right)$ samples cannot distinguish between $D_{yes}$ and $D_{no}$, and obtain Theorem~\ref{thm:non-sym}.

\color{black}
\section{Learning a distribution (Proof of Theorem~\ref{thm:nonconcub})}\label{sec:tol-ub}

In this section, we prove an upper bound related to the tolerant testing of more general properties.
Following a folklore result, when provided with oracle access to an unknown distribution $D$, we can always construct a distribution $D'$, such that the $\ell_1$ distance between $D'$ and $D$ (the unknown distribution) is at most $\eps$, by using $\Oh(\frac{n}{\eps^2})$ samples from $D$~\footnote{There is a writeup of this folklore result by Cannone~\cite{canonne2020learning}.}.
In this section, we provide a procedure that can be used for tolerant testing of properties, and in particular hints at how general tolerance gap bounds could be proved in the future. Our algorithm learns an unknown distribution approximately with high probability, adapting to the input, using as few samples as possible.
Specifically, we prove that given a distribution $D$, if there exists a subset $S \subseteq [n]$ which holds most of the total probability mass of $D$, then the distribution $D$ can be learnt using $\Oh(|S|)$ samples, even if the algorithm is unaware of $|S|$ in advance. Our result is formally stated as follows:

\begin{theo}[Theorem~\ref{thm:nonconcub} formalized]\label{theo:learn-main}
Let $D$ denote the unknown distribution over $\Omega=[n]$, and assume that there exists a set $S \subseteq [n]$ with $D(S) \geq 1- \frac{\eta}{2}$\footnote{Recall that the variation distance between two distribution is half than that of the $\ell_1$ distance between them. So, we take $D(S) \geq 1- \frac{\eta}{2}$ (with $\eta \in [0,2)$) instead of $D(S) \geq 1- \eta$ (with $\eta \in [0,1)$) .}, where $\eta \in [0,2)$ is known but $S$ and $\size{S}$ are unknown. Then there exists an algorithm that takes $\delta \in (0,2]$ as input and  constructs a distribution $D'$ satisfying $||D-D'||_1 \leq \eta +\delta$ with probability at least $\frac{2}{3}$. Moreover, the algorithm uses, in expectation, $\Oh\left(\frac{{\size{S}}}{\delta^2}\right)$ many samples from $D$.
\end{theo}

Note that in the above theorem, the algorithm has no prior knowledge of $\size{S}$. Before directly proving the above, we first show that if $\size{S}$ is known, then $\Oh(\size{S})$ many samples are enough to approximately learn the distribution $D$. We would like to point out that a similar question has been studied under the local differential privacy model with communication constraints, by Acharya, Kairouz, Liu and Sun~\cite{DBLP:conf/alt/AcharyaKLS21} and by Chen, Kairouz and {\"{O}}zg{\"{u}}r~\cite{DBLP:conf/nips/ChenKO20}.

\begin{lem}[Theorem~\ref{theo:learn-main} when $\size{S}$ is known]\label{lem:know}
Let $D$ be the unknown distribution over $\Omega=[n]$ such that there exists a set $S \subseteq [n]$ with $\size{S}=s$, and $\eta \in [0,2)$ such that $D(S) \geq 1- \frac{\eta}{2}$, where $s \in [n]$ and $\eta \in (0,1)$ are known. Then there exists an algorithm that takes $\delta \in (0,2]$ as an input and  constructs a distribution $D'$ satisfying $||D-D'||_1 \leq \eta +\delta$ with probability at least $\frac{9}{10}$. Moreover, the algorithm uses $\Oh\left(\frac{s}{\delta^2}\right)$ many samples from $D$.

\end{lem}

We note that Lemma~\ref{lem:know} can be obtained from the work of Acharya,
Diakonikolas, Li and Schmidt~\citep[Theorem 2]{acharya2017sample}. For completeness, we
give a self-contained proof for this lemma below.

We later adapt the algorithm of Lemma~\ref{lem:know} to give a proof to the scenario where $\size{S}$ is unknown, using a guessing technique. The idea is to guess $\size{S}=s$ starting from $s=1$, and then to query for $\Oh\left(s \right)$ many samples from the unknown distribution $D$. From the samples obtained, we construct a distribution $D_s$, and use Lemma~\ref{pro:val-tol} presented below to distinguish whether $D_s$ and $D$ are close or far. We argue that, for $s \geq \size{S}$, $D_s$ will be close to $D$ with probability at least $\frac{9}{10}$. We bound the total probability for the algorithm reporting a $D'$ that is too far from $D$ (for example when terminating before $s \geq |S|$), and also bound the probability of the algorithm not terminating in time when $s$ becomes at least as large as $|S|$.

\begin{lem}[\cite{ValiantV11}]
\label{pro:val-tol}
Let $D_u$ and $D_k$ denote unknown (input) and known (given in advance) distributions respectively over $\Omega=[n]$, such that the support of $D_u$ is a set of $s$ elements of $[n]$.  Then there exists an algorithm $\mbox{{\sc Tol-Alg}}(D_u,D_k,\eps_1,\eps_2,\kappa)$ that takes the full description of $D_k$, two proximity parameters $\eps_1,\eps_2$ with $0 \leq \eps_1 <\eps_2 \leq 2$ and $\kappa \in (0,1)$ as inputs, queries $\Oh\left(\frac{1}{(\eps_2 - \eps_1)^2}\frac{s}{\log s} \log \frac{1}{\kappa}\right)$ many samples from $D_u$, and distinguishes whether $||D_u-D_k||_1\leq \eps_1$ or  $||D_u-D_k||_1\geq \eps_2$ with probability at least $1-\kappa$~\footnote{The multiplicative factor $\log \frac{1}{\kappa}$ is for amplifying the success probability from $\frac{2}{3}$ to $1-\kappa$.}. 
\end{lem}

Note that Theorem~\ref{theo:learn-main} talks about learning a distribution with $\Oh(\size{S})$ samples, where there exists an unknown set $S$ with $D(S)\geq 1-\eta/2$. To prove  Theorem~\ref{theo:learn-main}, we use Lemma~\ref{pro:val-tol}, that crucially uses less than $s$ queries for tolerant identity testing (as opposed to learning).

The original bound following the paper of Valiant and Valiant~\cite{ValiantV11} is $\Oh\left(\frac{1}{(\eps_2-\eps_1)^2} \frac{n}{\log n}\right)$, which holds for any general distributions $D_u$ and $D_k$ with constant success probability. When deploying Lemma~\ref{pro:val-tol}, we ``contract'' the set $\Omega \setminus \supp(D_k)$ to a single element, which allows us to substitute $s+1$ for $n$. Note that this does not change the distance between $D_k$ and $D_u$. Hence, $\Oh\left(\frac{1}{(\eps_2-\eps_1)^2}\frac{s}{\log s}\right)$ samples from $D_u$ are enough for constant success probability. Following a recent work of Cannone, Jain, Kamath and Li~\cite{canonne2021price}, the dependence on the proximity parameters can be slightly improved. However we are not using that result since the focus of this work is different.

We first prove Lemma~\ref{lem:know}, and then proceed to prove Theorem~\ref{theo:learn-main}.

\begin{proof}[Proof of Lemma~\ref{lem:know}]
Let $Z$ be a multi-set of $\Oh\left(\frac{s}{\delta^2}\right)$ samples taken from $D$. The algorithm constructs a distribution $D':[n] \rightarrow [0,1]$ such that $$ D'(x)=\frac{\#~\mbox{times}~x~\mbox{appears in}~Z}{\size{Z}}.$$

Observe that $||D-D'||_1 = 2 \max\limits_{E \subseteq [n]} |D(E)-D'(E)|$. So, we will be done by showing the following:
\begin{equation}\label{Eqn:E}
    \mbox{With probability at least $\frac{9}{10}$, $|D(E)-D'(E)|\leq \frac{\eta+\delta}{2}$ for all $E \subseteq [n]$}
\end{equation}

Note that there are $2^n$ possibilities for $E$. So, a direct application of the union bound would require a failure probability of at most $\Oh(\frac{1}{2^n})$ for each $E$ not satisfying $|D(E)-D'(E)|\leq \frac{\eta+\delta}{2}$, that is, $\Oh(\frac{n}{\delta^2})$ samples would be needed. Assuming that $D$ is concentrated ($D(S)\geq 1- \frac{\eta}{2}$), we argue below that it is enough to have a failure probability of $\Oh(\frac{1}{2^s})$ for each $T$ not satisfying $|D(T)-D'(T)|\leq \frac{\delta}{4}$, but first we show that this is indeed the probability that we achieve.


\begin{obs}
Consider $T \subseteq [n]$. $|D(T)-D'(T)| \leq \frac{\delta}{4}$ holds with probability at least $1-\frac{1}{100 \cdot 2^s}$.
\end{obs}
\begin{proof}
Let $X_i$ denote the binary random variable that takes value $1$ if and only if the $i$-th sample in $Z$ is an element of $T$, where $i \in [\size{Z}]$.
So,  $D'(T)=\frac{1}{\size{Z}} \sum\limits_{i=1}^{\size{Z}}X_i.$

Observe that the expectation of $D'(T)$ is $\E\left[D'(T)\right]=D(T)$. Applying Chernoff bound (Lemma~\ref{lem:cher_bound2}), we get the desired result.
\end{proof}

By the above observation for every subset of $S$, applying the union bound over all possible subsets of $S$, we have $|D(T)-D'(T)|\leq \frac{\delta}{4}$ for every $T \subseteq S$ with probability at least $\frac{99}{100}$. Further applying the observation for $T=\Omega \setminus S$, we have $|D(\Omega \setminus S)-D'(\Omega \setminus S)|\leq \frac{\delta}{4}$ with probability at least $1-\frac{1}{100 \cdot 2^s}$. 
 
Let $\cE$ be the event that $|D(T)-D'(T)|\leq \frac{\delta}{4}$ for every $T \subseteq S$, and $|D(\Omega \setminus S)-D'(\Omega \setminus S)|\leq \frac{\delta}{4}$. Note that $\Pr(\cE) \geq \frac{9}{10}$. 
So, to prove Equation~\eqref{Eqn:E} and conclude the proof of Lemma~\ref{lem:know}, we show that $|D(E)-D'(E)| \leq \frac{\eta +\delta}{2}$ holds, in the conditional probability space when $\cE$ occurs, for any $E \subseteq [n]$.
 
\begin{eqnarray*}
 |D(E)-D'(E)| &\leq& |D(E\cap S)-D'(E \cap S)|+|D(E\cap (\Omega \setminus S))-D'(E \cap (\Omega \setminus S))| \\
&\leq& \frac{\delta}{4} + \max \ \{ {D( \Omega \setminus S), D'( \Omega \setminus S)} \}  \\
&\leq & \frac{\delta}{4} + D(\Omega \setminus S) + \frac{\delta}{4} \\
&\leq& \frac{\eta + \delta}{2}.
\end{eqnarray*}
\end{proof}

\begin{proof}[Proof of Theorem~\ref{theo:learn-main}]The algorithm is as follows:
\begin{enumerate}
    \item Set $s=1$.
    
    \item Query for a multi-set $Z_s$ of $\Oh\left(\frac{s}{\delta^2}\right)$ many samples from $D$.
    
    \item Construct a  distribution $D_s:[n] \rightarrow [0,1]$ such that 
$$ D_s(x)=\frac{\#~\mbox{times}~x~\mbox{appears in}~Z_s}{\size{Z_s}}$$

\item Call the algorithm $\mbox{{\sc Tol-Alg}}\left(D_s,D,\eta+ \frac{\delta}{2},\eta + \delta, \frac{1}{100 \log^2 s}\right)$  (corresponding to Lemma~\ref{pro:val-tol}) to distinguish whether $||D-D_s||_1\leq \eta+ \frac{\delta}{2}$ or $||D-D_s||_1 \geq  \eta+ \delta$. If we get $ ||D-D_s||_1\leq \eta + \frac{\delta}{2}$ as the output of {{\sc Tol-Alg}}, then we report $D'$ as the output  and {\sc Quit}.  Otherwise, we double the value of $s$. If $s \leq 2n$, go back to Step $2$. Otherwise, report {\sc Failure}~\footnote{By Lemma~\ref{pro:val-tol}, this step uses $\Oh(\frac{s}{\delta^2})$ samples.}. 
\end{enumerate}

Let $\cS$ denote the event that the algorithm quits with the desired output. We first show that $\Pr(S)\geq \frac{2}{3}$. Then we analyze the expected sample complexity of the algorithm.

Observe that the algorithm quits after an iteration with guess $s$ such that {\sc Alg-Tol} reports $ ||D-D_s||_1\leq \eta + \frac{\delta}{2}$. So, in that case,
the probability that the algorithm exits with an output not satisfying $||D-D_s||_1\leq \eta + \delta$ is at most $\frac{1}{100 \log^2 s}$. When summing this up over all possible $s$ (all powers of $k$, even up to infinity), the probability that the algorithm does not produce the desired output, given that it quits, is at most  $\sum\limits_{k=1}^{\infty} \frac{1}{100 k^2}\leq \frac{1}{10}$. So, denoting $\cQ$ as the event that the algorithm quits without reporting {\sc Failure}, $\Pr(\cS~|~\cQ) \geq \frac{9}{10}$.

For the lower bound on $\Pr(\cQ)$, consider the case where $s \geq \size{S}$. In this case, $||D_s-D||_1 \leq \eta + \frac{\delta}{2}$ with probability at least $\frac{9}{10}$, and {\sc Tol-Alg} quits by reporting $D_s$ as the output with probability at least $1-\frac{1}{100 \log^2 s}$.  So, for any guess $s \geq |S|$, the algorithm quits and reports the desired output with probability at least $\frac{4}{5}$. So, the probability that the algorithm quits without reporting failure is at least the probability that the algorithm quits with a desired output at some iteration with a guess $s \geq |S|$, which is at least $1-(\frac{1}{5})^{(\log n-\log |S|+1)}$. That is, $\Pr(\cQ)\geq \frac{4}{5}$.

Hence, the success probability of the algorithm can be lower-bounded as

$$\Pr(\cS)\geq \Pr(\cQ)\cdot \Pr(\cS~|~\cQ)\geq \frac{9}{10}\cdot \frac{4}{5}>\frac{2}{3}.$$

Now, we analyze the sample complexity of the algorithm. The algorithm queries for $\Oh(s)$ samples when it runs the iteration whose guess is $s$. The algorithm goes to the iteration with guess $s >\size{S}$ if all prior iterations which guessed more than $|S|$ failed, which holds with probability at most $\Oh\left((\frac{1}{5})^{\floor{\log s/\size{S}}}\right)$. Hence the expected sample complexity of the algorithm is at most

$$\sum \limits_{k:s=2^k<|S| } \Oh(\frac{s}{\delta^2}) +\sum \limits_{k:s=2^k \geq |S|}  \Oh\left(\left(\frac{1}{5}\right)^{\floor{\log (s/\size{S})}}\cdot \frac{s}{\delta^2} \right) = \Oh(\frac{\size{S}}{\delta^2}).$$

To explain the above equality, note that in the LHS of the above equation, each term of the second sum is bounded by $\Oh((\frac{1}{5})^{(k-\log |S|)} \cdot 2^{(k-\log |S|)} \cdot \frac{|S|}{\delta^2})$. Thus, substituting $k-\log(|S|)$ by $r$, we see that the second part of the LHS is upper bounded by $\sum\limits_{r\geq 0} \Oh\left((\frac{2}{5})^r \cdot \frac{|S|}{\delta^2}\right)$ which is clearly $\Oh(\frac{|S|}{\delta^2})$. Thus we have the above bound. 
\end{proof}

\bibliographystyle{alpha}
\bibliography{reference}

\end{document}